\documentclass[11pt,a4paper]{article}
\pdfoutput=1
\usepackage{jheppub}
\usepackage[T1]{fontenc}
\usepackage{lmodern,microtype}
\usepackage{amsmath,amssymb,mathtools,amsthm,booktabs}
\newtheorem{theorem}{Theorem}[section]
\newtheorem{proposition}[theorem]{Proposition}
\newtheorem{lemma}[theorem]{Lemma}
\newtheorem{corollary}[theorem]{Corollary}
\theoremstyle{definition}

\newtheorem{conjecture}[theorem]{Conjecture}
\DeclareMathOperator{\Irr}{Irr}
\DeclareMathOperator{\Hom}{Hom}
\DeclareMathOperator{\ord}{ord}
\DeclareMathOperator{\rank}{rank}
\DeclareMathOperator{\Rad}{Rad}
\DeclareMathOperator{\Tr}{Tr}
\newcommand{\Fp}{\mathbb F_p}
\newcommand{\U}{\mathrm U(1)}
\newcommand{\cC}{\mathcal C}
\newcommand{\cD}{\mathcal D}
\newcommand{\DW}{\mathrm{DW}}
\newcommand{\WZW}{\mathrm{WZW}}
\hypersetup{pdftitle={Maximal Total Quantum Dimension at Bounded Rank in Dijkgraaf-Witten Theories},pdfauthor={Ce Shen},pdfkeywords={Dijkgraaf-Witten theories, total quantum dimension, categorical rank, topological entanglement entropy}}
\title{Maximal Total Quantum Dimension at Bounded Rank in Dijkgraaf-Witten Theories}
\author[a]{Ce Shen}
\affiliation[a]{Beijing Institute of Mathematical Sciences and Applications (BIMSA),
Beijing, China}
\emailAdd{shence@bimsa.cn}
\abstract{We investigate the maximal total quantum dimension at a bounded categorical rank cutoff $R$ in three-dimensional Dijkgraaf--Witten theories. This optimization problem physically governs the maximal state-optimized torus topological entanglement entropy at a fixed torus ground-state degeneracy, while providing a quantitative constraint on the search space required to classify these topological order phases. We analyze how topological twists (3-cocycles) and gauge-group structures independently govern this trade-off. For elementary-Abelian gauge groups at a fixed prime $p$, we prove that introducing cocycle twists elevates the maximum total quantum dimension from square-root growth $\Theta_p(R^{1/2})$ to linear growth $\Theta_p(R)$ in the rank cutoff, driven by the geometric protection of cyclic holonomies. For untwisted theories, we establish sharp envelopes for symmetric and alternating gauge groups, and derive group-theoretic bounds across almost-simple, solvable, and radical-free families. In particular, almost-simple groups and their direct products obey a quadratic-logarithmic upper envelope: $\log\mathcal{D} = O((\log R)^2)$. Finally, we compare these scales against WZW models and discuss the status of the general quadratic-logarithmic envelope conjecture.}
\begin{document}
\maketitle

\section{Introduction}

Topological phases of matter in two spatial dimensions exhibit universal
forms of quantum order that transcend Landau's paradigm of local order
parameters and spontaneous symmetry breaking \cite{WenNiu,NayakEtAl}.
At zero temperature, these gapped quantum systems support fractionalized,
deconfined quasiparticle excitations known as anyons, possessing exotic
braiding statistics and a topological ground-state degeneracy that
depends on the underlying spatial manifold. From an entanglement
standpoint, intrinsic topological order is an immutable consequence of
long-range quantum entanglement, which cannot be created or destroyed
by shallow, local unitary quantum circuits \cite{ChenGuWen}.
Uncovering the mathematical and physical constraints that govern this
entanglement landscape is central to both fundamental quantum many-body
physics and the classification of topological quantum states.

At low energies, the algebraic structure of deconfined anyonic excitations
in $(2+1)$-dimensional bosonic topological phases is mathematically
formalized by a unitary modular tensor category (UMTC), denoted $\cC$
\cite{WenBosonic,NayakEtAl}. In this categorical framework, the simple
objects label the topologically distinct anyon types (or superselection
sectors), the fusion algebra dictates the multi-channel superposition
rules when anyons are brought together, and the associativity ($F$-symbols)
and braiding ($R$-symbols) maps characterize basis changes and spatial
exchanges of particles. Consistency of these physical operations is
guaranteed by the pentagon and hexagon polynomial equations.
Crucially, modularity embodies the nondegeneracy of topological braiding:
the vacuum sector is the unique anyon that braids trivially with all
excitations in the bulk. This structure provides a complete,
universal description of the bulk anyon sector, independent of microscopic
Hamiltonian details.

The same algebraic structures play a cornerstone role in two-dimensional
rational conformal field theory (RCFT). A rational chiral CFT possesses
finitely many primary fields with respect to its chiral algebra. The fusion
rules of these fields, together with the braiding and monodromy of conformal
blocks and the modular transformations of characters on the torus, mirror the
anyon braiding and modular data \cite{Verlinde,MooreSeiberg}. Under the usual
finiteness and regularity hypotheses ensuring modularity, the chiral
representation category is modular and, in the unitary setting, forms a
UMTC \cite{HuangModularity}. Nevertheless, a modular tensor category does not
uniquely specify a full conformal field theory: data such as the chiral
algebra, central charge, and conformal dimensions provide additional
physical information. In particular, tensoring a given CFT with a holomorphic
theory—which has trivial representation category—shifts the central charge
without altering the modular category \cite[Section 2.2]{TenerWang}.
Finiteness of categorical data at a given rank therefore does not imply
finiteness of full CFTs. Conversely, categorical modular data can serve as
generative input for constructing critical lattice models via competing anyon
condensation \cite{JiFactory}.

These intimate connections establish the classification of modular tensor
categories as a guiding problem for topological phases of matter and RCFT.
A natural grading for this classification is the \emph{categorical rank}
$r(\cC) = |\Irr(\cC)|$, which counts the number of distinct anyon species.
Physically, $r(\cC)$ represents the topological ground-state degeneracy
on the spatial torus $T^2$. At low rank, systematic classification programs
have mapped out allowed fusion rules and modular data, leveraging the modular
$S$- and $T$-matrices as strong numerical constraints \cite{RowellStongWang,NRWclassification}.
However, modular data alone do not provide a complete invariant: inequivalent
phases, such as twisted finite-group doubles, can share identical modular
matrices \cite{MignardSchauenburg}.

A central challenge in classifying topological phases at a bounded rank is the
potentially unbounded size of the underlying search space. Fixing the number
of anyon types $r(\cC) \le R$ does not, by itself, provide an explicit,
operational bound on how large the quantum dimensions $d_a$ or the integer
fusion multiplicities $N_{ab}^c$ can become. The rank-finiteness theorem of
Bruillard, Ng, Rowell, and Wang \cite[Theorem 3.1]{BNRW} proves that only
finitely many inequivalent UMTCs exist at any fixed rank. Their argument bounds
the order of the $T$ matrix and uses an $S$-unit equation to establish finiteness. However, they do not provide explicit bounds on the quantum dimensions or fusion multiplicities. The remaining challenge is therefore
obtaining sharp,
physically informative envelopes and manageable bounds for numerical or
combinatorial classification searches.

The \emph{total quantum dimension} $\cD(\cC) = \sqrt{\sum_a d_a^2}$ provides
an ideal physical and mathematical handle for this problem. While the rank
$r(\cC)$ counts the number of sectors, $\cD(\cC)$ captures the cumulative
topological weight of the phase, heavily weighting non-Abelian anyons
($d_a > 1$) over Abelian excitations ($d_a = 1$). Because quantum dimensions
satisfy $d_a d_b = \sum_c N_{ab}^c d_c$, a universal upper bound on $\cD$
immediately controls every individual quantum dimension $d_a$ and imposes
a finite ceiling on all integer fusion multiplicities:
\begin{equation}
 0 \le N_{ab}^c \le \frac{d_a d_b}{d_c} \le \cD^2.
 \label{eq:intro-fusion-cutoff}
\end{equation}
By Ocneanu rigidity \cite[Theorem 2.28]{ENO}\cite[Theorem 3.11]{BNRW},
bounding $N_{ab}^c$ restricts the classification to finitely many candidate
fusion rings, each admitting only finitely many associator and braiding
solutions up to equivalence.

This leads directly to the extremal question investigated here:
\emph{What is the maximal total quantum dimension that a topological phase can
support at a bounded categorical rank cutoff $R$?}
To accommodate families whose ranks form sparse discrete sets, we formulate
this optimization across all theories below an external cutoff $R$. Following
Ref.~\cite{ShenWZW}, we consider the logarithmic envelope
\begin{equation}
 F(R) = \sup_{\substack{\cC\ {\rm UMTC}\\r(\cC)\le R}} 2\log\cD(\cC).
 \label{eq:intro-UMTC-envelope}
\end{equation}
The quantity $F(R)$ admits a direct physical interpretation in terms of quantum
entanglement. On a spatial torus partitioned into two cylinders by two circular
boundaries of lengths $L_1$ and $L_2$, the universal topological entanglement
entropy (TEE) contains a constant sub-leading correction $\Gamma_{T^2}(\psi)$
that depends on the chosen ground state $|\psi\rangle$ \cite{DongEtAl,ZhangEtAl}.
Optimizing over all normalized ground states in the torus Hilbert space reveals
that the maximum TEE is attained by any definite Abelian flux state, including
the vacuum, yielding
$\max_\psi \Gamma_{T^2}(\psi) = 2\log \cD(\cC)$ \cite{ShenWZW}.
Thus, $F(R)$ quantifies the maximal state-optimized torus TEE compatible with a
given bound on the torus ground-state degeneracy.

While bounding $F(R)$ across the entire universe of UMTCs is an ambitious
frontier, tractable and physically prominent families of theories allow this
trade-off to be examined systematically. In Ref.~\cite{ShenWZW}, the envelope
was analyzed for WZW modular tensor categories, uncovering a sharp
quadratic-logarithmic growth $F_{\WZW}(R) \sim \frac{7\zeta(3)}{4\pi^2} (\log_2 R)^2$.
In the present work, we turn to $(2+1)$-dimensional Dijkgraaf--Witten (DW)
theories—discrete gauge theories specified by a finite gauge group $G$ and
a topological action (twist) defined by a 3-cocycle class $[\omega] \in H^3(G,\U)$
\cite{DW,DPR}. The anyonic excitations of a DW theory are described by the
modular representation category of the twisted quantum double:
\begin{equation}
 \cC(G,\omega) = \operatorname{Rep}D^\omega(G).
 \label{eq:DWcategory}
\end{equation}
A remarkable feature of DW theories is that their total quantum dimension is
strictly determined by the gauge group order, $\cD(G,\omega) = |G|$, completely
independent of the twist $\omega$. In contrast, the categorical rank $r_\omega(G)$—the
torus ground-state degeneracy—is sensitive to both the algebraic
structure of $G$ and the topological cocycle $\omega$. The central problem
therefore becomes an extremal group-theoretic optimization: maximizing the
group order $|G|$ subject to a rank ceiling $r_\omega(G) \le R$. We define the
corresponding envelopes for DW theories:
\begin{equation}
 F_{\DW}(R) = \sup_{\substack{G\ {\rm finite},\,[\omega]\in H^3(G,\U)\\
 r_\omega(G)\le R}} 2\log|G|,
 \qquad
 F_{\DW}^{(1)}(R) = \sup_{\substack{G\ {\rm finite}\\
 r_1(G)\le R}} 2\log|G|,
\end{equation}
where $F_{\DW}^{(1)}$ denotes the untwisted ($\omega = 1$) envelope.

This optimization reveals two distinct physical mechanisms:
\begin{enumerate}
 \item \emph{Cocycle twisting at fixed gauge group.}
 For a fixed group \(G\), introducing a nontrivial \(3\)-cocycle twist \(\omega\) preserves the total quantum dimension \(\mathcal D=|G|\), but can substantially reduce the number of anyon types. For each magnetic flux \(g\), the charge labels become irreducible projective representations of the centralizer \(C_G(g)\), with 2-cocycle \(\alpha_g\) of the centralizer \(C_G(g)\) induced by \(\omega\). Consequently, only the \(\alpha_g\)-regular conjugacy classes of \(C_G(g)\) contribute to the projective character theory, and the number of charge sectors over a given flux can decrease. Every magnetic-flux conjugacy class remains present.
 For elementary-Abelian gauge groups $G = (\mathbb{Z}_p)^n$ at a fixed prime $p$,
 we prove that optimal 3-cocycles transform the asymptotic envelope from
 \begin{equation}
  F_p^{(1)}(R) = \log R + O_p(1) \; (\text{untwisted})
  \quad \Longrightarrow \quad
  F_p(R) = 2\log R + O_p(1) \; (\text{twisted}).
  \label{eq:intro-laws}
 \end{equation}
 Correspondingly, the maximal total quantum dimension changes from
 square-root growth, $\cD_{\max,p}^{(1)}(R) = \Theta_p(R^{1/2})$,
 to linear growth, $\cD_{\max,p}(R) = \Theta_p(R)$.
 The mechanism preventing further compression is the geometric phenomenon of
 \emph{cyclic protection}: commuting holonomies that generate cyclic subgroups
 cannot acquire non-trivial topological phases under any twist, establishing a
 universal floor on the anyon count.
 \item \emph{Gauge-group architecture.}
 In the untwisted setting, the trade-off is governed entirely by group structure.
 For permutation gauge groups—the symmetric groups $S_n$ and alternating
 groups $A_n$—we establish the sharp asymptotic envelope
 \begin{equation}
  F_S^{(1)}(R) \sim F_A^{(1)}(R) \sim A_S (\log R)^{3/2}\log\log R,
  \qquad A_S = \frac{2}{\sqrt{3\zeta(2)\zeta(3)}},
  \label{eq:intro-permutation}
 \end{equation}
 via an asymptotic inversion of the commuting-triple generating function \cite{BF,BFH}.
 Furthermore, we prove that almost-simple gauge groups and their finite direct
 products obey an upper envelope $F_{\mathrm{a.s.}}^{(1)}(R) \le C(\log R)^2$.
 We also establish structural bounds across solvable, radical-free, and general
 finite groups (summarized in Table~\ref{tab:envelopes}), and analyze how group
 extensions and quotient losses restrict the synthesis of extremal theories.
\end{enumerate}

The charge--flux description and projective-character rank formula are
standard foundations \cite{DPR,HW,Willerton}. Our analysis uses them to
extract a cyclic protection index and an exact first-moment identity, and
to translate group-theoretic estimates into rank-cutoff envelopes.
The finite-geometric constructions \cite{DSsingular,DSforms}, permutation
enumerations \cite{BF,BFH,FS}, and class-number bounds used below are
credited at their points of use. This separates the optimization results
from the established inputs on which they depend.

The remainder of this paper is structured as follows.
Section~\ref{sec:setup} reviews the topological and quantum-double foundations,
details the charge--flux decomposition and torus entanglement optimization, and
establishes the family-by-family taxonomy in Table~\ref{tab:envelopes}.
Section~\ref{sec:twisting} focuses on the physics of topological twists: we
derive the cyclic protection index $\Xi(G)$, map the rank minimization over
elementary-Abelian groups to the geometry of alternating trilinear forms, and
prove the sharp linear envelope $\cD_{\max,p}(R) = \Theta_p(R)$ via the first-moment
method.
Section~\ref{sec:structure} explores non-Abelian gauge architectures, proving the
sharp permutation envelopes, the almost-simple quadratic-logarithmic bound,
quotient and central-extension constraints, and solvable/radical-free refinements.
Section~\ref{sec:outlook} contrasts these findings with chiral WZW models,
formulates the universal quadratic-logarithmic envelope conjecture, and highlights
open directions.
Appendix~\ref{app:technical} provides technical derivations for the alternating
phase, rank censuses, unitriangular groups, affine actions, wreath products, and
envelope variational identities.

\section{Topological Setup and the DW Cutoff Problem}
\label{sec:setup}

We begin by establishing the topological and quantum-double framework
governing our optimization, connecting bulk categorical invariants with the
geometry of torus entanglement. We then detail the charge--flux decomposition
of Dijkgraaf--Witten theories, showing how the rank-cutoff problem maps onto
an extremal question in finite group theory and cohomology.

\subsection{Rank, total quantum dimension, and torus entanglement}
\label{subsec:physical-background}

For a $(2+1)$-dimensional topological phase described by a UMTC $\cC$,
the quantities needed here are determined by its simple superselection
sectors $\Irr(\cC)$ and quantum dimensions $d_a$, which can be read from
the unitary modular data $S$ and $T$ \cite{WenBosonic,NayakEtAl}.
These data need not determine the full category.
Here, $0 \in \Irr(\cC)$ denotes the vacuum sector.
The categorical rank $r(\cC) = |\Irr(\cC)|$ counts the distinct anyon species,
while the total quantum dimension $\cD(\cC)$ aggregates their collective weights:
\begin{equation}
 r(\cC) = |\Irr(\cC)|, \qquad
 \cD(\cC)^2 = \sum_{a \in \Irr(\cC)} d_a^2, \qquad
 d_a = \frac{S_{0a}}{S_{00}}, \qquad
 S_{00} = \frac{1}{\cD(\cC)}.
 \label{eq:physical-invariants}
\end{equation}
In topological quantum field theory (TQFT), the categorical rank and total
quantum dimension govern the partition functions on closed 3-manifolds:
\begin{equation}
 Z_{\cC}(T^3) = \dim\mathcal{H}_{\cC}(T^2) = r(\cC),
 \qquad
 Z_{\cC}(S^3) = \frac{1}{\cD(\cC)}.
 \label{eq:partition-invariants}
\end{equation}
Imposing an upper bound on categorical rank $r(\cC) \le R$ physically bounds
the topological ground-state degeneracy (GSD) on a spatial 2-torus.
Simultaneously, maximizing $\cD(\cC)$ corresponds to minimizing the vacuum
amplitude on $S^3$.

The total quantum dimension also admits a direct information-theoretic
interpretation through the topological entanglement entropy (TEE)
\cite{KitaevPreskill,LevinWen}. Consider a spatial 2-torus partitioned into two
cylindrical subregions $A$ and $\bar{A}$ by two circular cuts of lengths $L_1$
and $L_2$, as illustrated in figure~\ref{fig:torus-bipartition}.
The torus Hilbert space $\mathcal{H}_{\cC}(T^2)$ admits a canonical basis
$\{|a\rangle\}_{a \in \Irr(\cC)}$ corresponding to threading definite anyon
flux $a$ along the non-contractible cycle threading the cylinders.
For an arbitrary normalized ground state $|\psi\rangle = \sum_a \psi_a |a\rangle$,
with occupation probabilities $p_a = |\psi_a|^2$, the von Neumann entanglement
entropy in the topological scaling limit takes the form
$S_A = \alpha(L_1 + L_2) - \Gamma_{T^2}(\psi) + o(1)$, where $\alpha$ is a
non-universal boundary-area coefficient \cite{DongEtAl,ZhangEtAl}.
The universal subleading correction is
\begin{equation}
 \Gamma_{T^2}(\psi)
 = 2\log\cD -
 \left(2\sum_a p_a\log d_a - \sum_a p_a\log p_a\right).
 \label{eq:torus-tee}
\end{equation}
Because unitarity dictates $d_a \ge 1$ and the Shannon entropy $-\sum_a p_a \log p_a$
is non-negative, the bracketed expression is non-negative.
Maximizing the torus TEE across all normalized states in the torus ground-state
manifold yields
\begin{equation}
 \max_{\substack{\psi \in \mathcal{H}_{\cC}(T^2)\\\|\psi\|=1}}
 \Gamma_{T^2}(\psi) = 2\log\cD(\cC).
 \label{eq:state-maximum}
\end{equation}
This state-optimized torus TEE is saturated precisely when the system resides in
a definite Abelian flux state; the vacuum-flux state $|\psi\rangle = |0\rangle$
serves as a canonical maximizer. The prefactor of two reflects the presence of
two disjoint circular boundaries in the cylindrical bipartition.
(In contrast, a planar disk bipartition hosts a single boundary, giving
$\Gamma_{\rm disk} = \log \cD$ \cite{KitaevPreskill,LevinWen}.)
Complementary treatments of torus entanglement use string-net wavefunctions
\cite{LuoHuWuTorus}, edge theories \cite{WenMatsuuraRyu}, and torus-knot
bipartitions \cite{LoChang}. We keep the bipartition and entropy prescription
fixed; choices of subsystem algebra in lattice gauge theory constitute a
separate issue \cite{HungWanGaugeEntropy}.

The entanglement cuts here are not physical gapped boundaries. The latter
require additional boundary data, as made explicit in the Hamiltonian
constructions of Refs.~\cite{BullivantHuWanBoundary,HuEtAlBoundary}.
Anyon condensation provides a complementary description of such boundaries
and phase transitions \cite{BaisSlingerland,KongAnyon,KongErratum,LouShenChenHung}.
It can modify entanglement in systems with physical boundaries or defects
\cite{ChenHungLiWan,LouShenHungI,ShenLouHungII,ShenHungDefect}; the optimization
in \eqref{eq:state-maximum}, by contrast, concerns a closed spatial manifold.

\begin{figure}[htbp]
 \centering
 \includegraphics[width=0.75\linewidth]{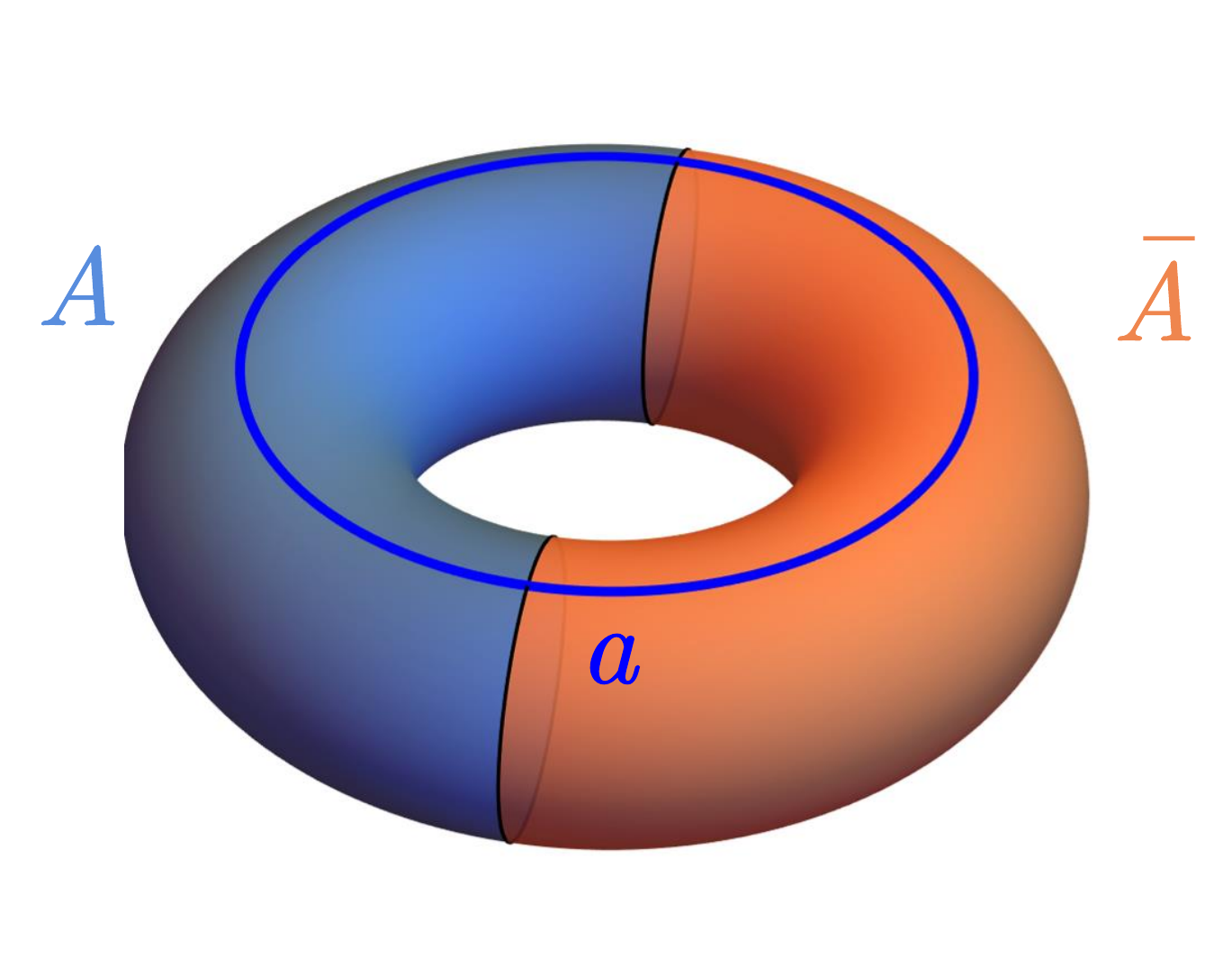}
 \caption{Cylindrical bipartition of a spatial 2-torus into regions $A$ and $\bar{A}$.
 The entanglement cut consists of two disjoint circles of lengths $L_1$ and $L_2$.
 The state $|a\rangle$ corresponds to threading an anyon flux label $a$ through the
 cylinder, giving rise to the universal TEE in \eqref{eq:torus-tee}.
 (Reproduced from Ref.~\cite{ShenWZW}.)}
 \label{fig:torus-bipartition}
\end{figure}

For any specified collection $\mathfrak{F}$ of topological phases, the maximal
achievable total quantum dimension and its corresponding logarithmic envelope
are defined by
\begin{equation}
 \cD_{\max,\mathfrak{F}}(R)
 = \sup_{\substack{\cC \in \mathfrak{F}\\r(\cC) \le R}} \cD(\cC),
 \qquad
 F_{\mathfrak{F}}(R) = 2\log\cD_{\max,\mathfrak{F}}(R).
 \label{eq:family-envelope}
\end{equation}
Thus, $F_{\mathfrak{F}}(R)$ measures the maximal state-optimized torus TEE
attainable within the family $\mathfrak{F}$ under a prescribed torus GSD bound $R$.
Throughout this paper, $R \ge 1$, all logarithms are natural unless explicitly
indexed, and $q_R = \log_2 R$ denotes the binary logarithm.

\subsection{Finite gauge theory and charge--flux sectors}
\label{subsec:DW-background}

In $(2+1)$-dimensional Dijkgraaf--Witten theory, the dynamical gauge fields are
flat connections with values in a finite group $G$ \cite{DW,Willerton}.
A flat connection on a 2-manifold assigns a group holonomy to each closed loop,
and gauge transformations act by overall conjugation.
The topological action is specified by a normalized 3-cocycle $\omega: G^3 \to \U$,
representing a cohomology class $[\omega] \in H^3(G,\U)$.
The 3-cocycle satisfies the non-Abelian pentagon condition
\begin{equation}
 \omega(b,c,d)\,\omega(a,bc,d)\,\omega(a,b,c)
 = \omega(ab,c,d)\,\omega(a,b,cd),
 \label{eq:background-cocycle}
\end{equation}
along with the normalization $\omega(a,b,c) = 1$ whenever any argument is the
group identity $1 \in G$.
Shifting $\omega$ by a 3-coboundary $\delta\beta$ merely corresponds to a gauge-invariant
rephasing of the path integral, leaving physical observables invariant.
In the Hamiltonian formulation, these phases are realized by Kitaev quantum double
models \cite{KitaevQuantumDouble} and their cocycle-twisted generalizations
\cite{HW,BuerschaperTwisted}. The related string-net framework
\cite{LevinWenStringNet} describes doubled topological phases, with the full
charge and dyon excitation spectrum accessible in extended Levin--Wen models
\cite{HuGeerWuDyons}.

\subsubsection{Charge--flux labels and quantum dimensions}
The simple anyonic excitations in a DW theory are classified by gauge-invariant
pairs of magnetic flux and electric charge \cite{DPR,HW}.
A magnetic flux sector is labeled by a conjugacy class $[g] = \{xgx^{-1} : x \in G\}$.
Fixing a representative flux $g \in G$, the residual gauge symmetry preserving this
holonomy is the centralizer subgroup $C_G(g) = \{x \in G : xg = gx\}$.
The electric charge carried by the excitation is an irreducible representation
$\pi$ of $C_G(g)$, generally \emph{projective} in the presence of a twist
\cite{DPR,HW}, satisfying
\begin{equation}
 \pi(x)\pi(y) = \alpha_g(x,y)\,\pi(xy), \qquad x,y \in C_G(g),
 \label{eq:background-projective}
\end{equation}
where the centralizer 2-cocycle $\alpha_g \in Z^2(C_G(g),\U)$ is induced by $\omega$:
\begin{equation}
 \alpha_g(x,y) = \frac{\omega(g,x,y)\,\omega(x,y,g)}{\omega(x,g,y)}.
 \label{eq:alpha}
\end{equation}
Ordinary representations occur when the multiplier is trivial; in particular,
$\alpha_1=1$ in the identity flux sector for every normalized $\omega$.
The anyons of the twisted quantum double $\cC(G,\omega) = \operatorname{Rep}D^\omega(G)$
are therefore indexed by pairs $([g],\pi)$, where $\pi \in \Irr_{\alpha_g}(C_G(g))$
is an irreducible $\alpha_g$-projective representation.
The quantum dimension of the sector $([g],\pi)$ is given by the product of its
magnetic orbit size and electric representation dimension:
\begin{equation}
 d_{([g],\pi)} = |[g]|\dim\pi = [G : C_G(g)]\dim\pi.
 \label{eq:anyon-dimensions}
\end{equation}
Because the sum of squared dimensions of all irreducible $\alpha_g$-projective
representations of $C_G(g)$ equals the group order $|C_G(g)|$ \cite{Karpilovsky},
the total quantum dimension evaluates to
\begin{equation}
 \cD(G,\omega)^2
 = \sum_{[g]} |[g]|^2 \sum_{\pi} (\dim\pi)^2
 = \sum_{[g]} |[g]|^2 |C_G(g)|
 = |G| \sum_{[g]} |[g]|
 = |G|^2.
 \label{eq:DWdimension}
\end{equation}
Hence, $\cD(G,\omega) = |G|$ is an exact topological invariant of the gauge group,
entirely independent of the twist $\omega$.
Twisting can change the number and dimensions of projective charges in each
flux sector, while keeping $\cD$ fixed. Each twisted centralizer algebra is
nonzero and semisimple, so every flux conjugacy class supports at least one
charge sector. Explicit modular-data formulas and computations for these
categories are developed in Refs.~\cite{CosteGannonRuelle,GruenMorrison}.

\subsubsection{Torus states and commuting holonomies}
The relation between flat connections and ground states follows from the
Hamiltonian realization of the untwisted theory ($\omega=1$). On a lattice,
each oriented edge carries a group element, and the ordered product along
a closed loop is its holonomy. The quantum-double Hamiltonian can be written
as a sum of mutually commuting projectors,
\[
 H=\sum_v(1-A_v)+\sum_f(1-B_f).
\]
Here $B_f$ projects onto trivial holonomy around the face $f$, while $A_v$
averages over gauge transformations at the vertex $v$. A ground state
satisfies $B_f=1$ and $A_v=1$ everywhere: it has no local magnetic or electric
excitations. Thus flatness alone is not sufficient; the wavefunction must
also be gauge invariant \cite{KitaevQuantumDouble,HW}.

Flatness means that every contractible loop has identity holonomy, but
noncontractible loops can retain nontrivial global holonomies. On the
spatial torus $T^2$, choose its two fundamental cycles with a common
basepoint and denote their holonomies by $g$ and $h$. Their commutator loop
is contractible, so
\[
 ghg^{-1}h^{-1}=1,\qquad gh=hg.
\]
After fixing the local gauge freedom, flat configurations are therefore
represented by commuting pairs $(g,h)$. The residual gauge transformation
at the basepoint acts by simultaneous conjugation,
\[
 U_x|g,h\rangle=|xgx^{-1},xhx^{-1}\rangle,\qquad x\in G.
\]
An orbit is the set of all pairs related by this action; its members are
gauge-equivalent representatives of the same configuration.

In this reduced description, a wavefunction has amplitudes $\psi(g,h)$ on
commuting pairs. Gauge invariance imposes
\[
 \psi(xgx^{-1},xhx^{-1})=\psi(g,h).
\]
The amplitude is therefore constant on each orbit, leaving exactly one
independent amplitude per orbit. More explicitly, with an orthonormal
configuration basis, each orbit $\mathcal O$ gives the invariant state
\[
 |\mathcal O\rangle=\frac{1}{\sqrt{|\mathcal O|}}
 \sum_{(g,h)\in\mathcal O}|g,h\rangle.
\]
These states are mutually orthogonal and span the ground-state space.
Consequently, the torus ground-state degeneracy is the number of
simultaneous conjugacy orbits of commuting pairs, rather than the number
of individual pairs \cite{HW,Willerton}.

To count the orbits, first fix a representative $g$ of each conjugacy
class $[g]$. The second holonomy lies in $C_G(g)$, and the remaining gauge
transformations conjugate it within $C_G(g)$. There are therefore
$k(C_G(g))$ choices, giving
\begin{equation}
 r_1(G) = \sum_{[g]} k(C_G(g)),
 \label{eq:chargeflux-physical}
\end{equation}
where $k(H)$ denotes the number of ordinary conjugacy classes (and ordinary
irreducible representations) of a finite group $H$.
This also equals the number of anyon labels $([g],\pi)$ described above.
The orbit basis and the definite-anyon-flux basis are different bases of
the same torus ground-state space, related by a character transform;
an individual orbit state need not have definite anyon flux.
Applying Burnside's lemma to the simultaneous conjugation action on commuting pairs
introduces a third commuting holonomy:
\begin{equation}
 h_3(G) := |\Hom(\mathbb{Z}^3,G)|, \qquad
 r_1(G) = \frac{h_3(G)}{|G|}, \qquad
 P_3(G) := \frac{h_3(G)}{|G|^3}.
 \label{eq:definitions}
\end{equation}
Here, $h_3(G)$ counts ordered commuting triples $(g,h,x) \in G^3$, and $P_3(G)$ is
the probability that three uniformly chosen group elements mutually commute.
In an untwisted theory, $r_1(G)$ is the torus GSD.
In a twisted theory, gauge transformations acquire cocycle-dependent phases,
so the equal-amplitude orbit construction need not survive the gauge projection;
the precise criterion for an orbit to survive is derived in Section~\ref{sec:twisting}.

To contrast the distribution of quantum dimensions, consider the non-Abelian
gauge group $S_3$ in the untwisted theory. The conjugacy classes are the identity
$[1]$ with $C_{S_3}(1) \cong S_3$, transpositions $[(12)]$ with $C_{S_3}((12)) \cong C_2$, and 3-cycles $[(123)]$ with $C_{S_3}((123)) \cong C_3$.
Their class numbers are $k(S_3) = 3$, $k(C_2) = 2$, and $k(C_3) = 3$, giving
$r_1(S_3) = 3 + 2 + 3 = 8$ anyon sectors, with total quantum dimension $\cD(S_3,1) = 6$.
The individual quantum dimensions are
\begin{equation}
 (d_a) = (\underbrace{1, 1, 2}_{\text{flux } 1}; \;
          \underbrace{3, 3}_{\text{flux } (12)}; \;
          \underbrace{2, 2, 2}_{\text{flux } (123)}),
 \qquad \sum_a d_a^2 = 36.
\end{equation}
In comparison, an Abelian gauge group of the same order, such as $C_6$, possesses
$r_1(C_6) = 36$ strictly Abelian sectors ($d_a = 1$), yielding the identical
$\cD = 6$. Non-Abelian gauge structures thus concentrate a large total quantum
dimension into significantly fewer, non-Abelian anyon species.

\subsection{Cutoff envelopes and the landscape of gauge families}
\label{subsec:cutoff-background}

The rigidity $\cD(G,\omega) = |G|$ transforms our physical optimization problem
into bounding group order against categorical rank.
The unrestricted Dijkgraaf--Witten envelope and the fixed-prime elementary-Abelian
envelope are given respectively by
\begin{align}
 F_{\DW}(R)
 &= \sup_{\substack{G\ {\rm finite},\,[\omega]\in H^3(G,\U)\\
                   r_\omega(G)\le R}} 2\log|G|,
 \label{eq:DWenvelope}\\
 F_p(R)
 &= \sup_{\substack{n\ge 0,\,[\omega]\in H^3(\Fp^n,\U)\\
                   r_\omega(\Fp^n)\le R}} 2n\log p.
 \label{eq:penvelope}
\end{align}
Restricting to trivial twists ($\omega = 1$) defines the untwisted envelopes
$F_{\DW}^{(1)}(R)$ and $F_p^{(1)}(R)$.
Importantly, our global finiteness theorem (Theorem~\ref{thm:global}) establishes
that bounding rank $r_\omega(G) \le R$ restricts $|G|$ to a finite range.
Because there are only finitely many groups of bounded order and $H^3(G,\U)$ is finite
for every finite group, the supremum in \eqref{eq:DWenvelope} is always attained as
a true maximum.
This is a finite set of presentations $(G,[\omega])$, not a unique labeling
of UMTCs: different presentations can give braided-equivalent categories.
Moreover, an asymptotic bound with an unspecified constant is not by itself
an executable numerical cutoff for enumeration.

To probe the structural origins of the envelope, we systematically evaluate
$F_{\mathfrak{G}}^{(1)}(R)$ and $F_{\mathfrak{G}}(R)$ across specific algebraic
classes $\mathfrak{G}$:
\begin{itemize}
 \item \emph{Elementary-Abelian groups} $(\mathbb{Z}_p)^n$: Linear vector spaces
 over $\mathbb{F}_p$, where cocycle twists map directly to alternating trilinear forms.
 \item \emph{Permutation groups} $S_n$ and $A_n$: Combinatorial gauge groups whose
 commuting holonomies correspond to $\mathbb{Z}^2$-permutation sets.
 \item \emph{Almost-simple groups}: Finite groups $S \le G \le \operatorname{Aut}(S)$
 where $S$ is a non-Abelian simple group, along with their direct products.
 \item \emph{Solvable groups}: Groups constructed by successive Abelian extensions,
 parameterized by their derived length $d$.
 \item \emph{Radical-free groups}: Groups possessing no non-trivial solvable normal
 subgroups ($\mathrm{Rad}(G) = 1$).
\end{itemize}

Table~\ref{tab:envelopes} compiles our rigorous envelope bounds across these families.
A two-sided envelope bound requires two complementary analyses: establishing a
\emph{lower bound} on rank in terms of group order (which produces an \emph{upper}
envelope on $F(R)$), and exhibiting an explicit sequence of theories with controlled
rank (which provides a realizable \emph{lower} envelope).
In the subsequent sections, we unfold the physical and algebraic mechanisms
underlying each row in this hierarchy.

\begin{table}[t]
\centering
\small
\renewcommand{\arraystretch}{1.15}
\caption{Rigorous envelopes for maximal $2\log\cD$ at categorical rank cutoff $R$
as $R\to\infty$. Fixed parameters ($p$, $d$, or $\varepsilon$) are indicated in
the asymptotic notation. All bounds are derived from the specified theorems and
external inputs cited in the text.}
\label{tab:envelopes}
\begin{tabular}{@{}p{.34\linewidth}p{.43\linewidth}p{.17\linewidth}@{}}
\toprule
Gauge-group class and twists & Envelope $F(R) = 2\log\cD_{\max}(R)$ & Source\\
\midrule
$(\mathbb Z_p)^n$, fixed $p$, any twist
& $F_p(R) = 2\log R + O_p(1)$ & Cor.~\ref{cor:envelope}\\
$(\mathbb Z_p)^n$, fixed $p$, untwisted
& $F_p^{(1)}(R) = \log R + O_p(1)$ & Cor.~\ref{cor:envelope}\\[3pt]
$S_n$ or $A_n$, untwisted
& $F_S^{(1)}(R), F_A^{(1)}(R) \sim A_S(\log R)^{3/2}\log\log R$
& Cor.~\ref{cor:symmetric}\\
Almost-simple \& direct products, untwisted
& $F_{\mathrm{a.s.}}^{(1)}(R) \le C(\log R)^2$
& Thm.~\ref{thm:almostsimple}\\[3pt]
Derived length $\le d$, untwisted
& $F_{\mathrm{dl}\le d}^{(1)}(R) \le 2\alpha_d^{-1}\log R$
& Prop.~\ref{prop:derived}\\
Nilpotency class $\le 2$, untwisted
& $F_{\mathrm{class}\le 2}^{(1)}(R) \le 2\log R$
& Prop.~\ref{prop:derived}\\[3pt]
Solvable, untwisted
& $F_{\mathrm{solv}}^{(1)}(R) = O(R\log\log R)$
& Thm.~\ref{thm:solvable}\\
Solvable, any twist
& $F_{\mathrm{solv}}(R) = O(R\log R)$
& Thm.~\ref{thm:solvable}\\
$\mathrm{Rad}(G)=1$, any twist
& $F_{\mathrm{Rad}=1}(R) = O_\varepsilon((\log R)^{3+\varepsilon})$
& Cor.~\ref{cor:radicalfree}\\[3pt]
All groups, untwisted
& $F_{\DW}^{(1)}(R) = O_\varepsilon(R(\log R)^{2+\varepsilon})$
& Thm.~\ref{thm:global}\\
All groups, any twist
& $F_{\DW}(R) = O_\varepsilon(R(\log R)^{3+\varepsilon})$
& Thm.~\ref{thm:global}\\
\bottomrule
\end{tabular}
\end{table}

The constants in the fixed-derived-length row are
\[
 \alpha_1=2,\qquad
 \alpha_d=(3\cdot2^{d-2}-1)^{-1}\quad(d\geq2).
\]
Neither solvability without a derived-length cutoff nor the absence of a
solvable radical currently yields the desired quadratic-logarithmic
bound by these arguments.

There are two directions to keep separate. A \emph{lower} bound on
rank in terms of $|G|$ limits the group orders admitted at a fixed $R$
and hence gives an \emph{upper} envelope bound.
Conversely, an explicit family with controlled rank supplies
admissible theories and therefore a \emph{lower} envelope bound.
A sharp restricted law needs both directions, together with control
of the gaps between successive admissible ranks.

Table~\ref{tab:envelopes} is a map of the results, not a list of
saturating constructions. The elementary-Abelian and permutation rows
give sharp restricted scales. The almost-simple row gives only an
upper bound. The final rows establish finiteness for all finite
gauge groups, but remain far above the conjectured $(\log R)^2$ scale.
The two following sections develop, respectively, the twisting and
group-structure mechanisms behind this comparison.

\section{Cocycle Twisting and Elementary-Abelian Envelopes}
\label{sec:twisting}
\label{sec:protected}

At a fixed gauge group $G$, introducing a 3-cocycle twist $\omega$ preserves
the total quantum dimension $\cD(G,\omega) = |G|$, but can drastically reduce
the categorical rank by changing projective charge multiplicities and projecting
out nonregular commuting-pair orbits. This compares distinct topological theories,
not a perturbative splitting of degeneracies within one phase.
In this section, we analyze the fundamental limits of this compression mechanism.
We first identify a topological obstruction—the phenomenon of \emph{cyclic protection}—which
guarantees that certain commuting holonomies survive every possible twist,
imposing a rigorous lower bound on anyon count.
We then specialize to elementary-Abelian gauge groups $V = \mathbb{F}_p^n$, where
the minimization of rank over all cocycles maps onto the geometry of alternating
trilinear forms and singular projective lines.
Using a probabilistic first-moment analysis over the space of 3-cocycles, we prove
that optimal twists achieve linear scaling in the rank cutoff, establishing the
sharp fixed-prime envelopes
\begin{equation}
 (p+1)p^n - p \le m_p(n) < (p^2+p+1)p^n \quad (n\ge 3),
 \qquad r_1(\Fp^n) = p^{2n},
 \label{eq:family-summary}
\end{equation}
where $m_p(n) := \min_{[\omega]} r_\omega(\Fp^n)$.
Consequently, twisting elevates the maximal total quantum dimension from
$\cD_{\max,p}^{(1)}(R) = \Theta_p(R^{1/2})$ to $\cD_{\max,p}(R) = \Theta_p(R)$.

\subsection{Regular commuting pairs and cyclic protection}

The rank of a twisted Dijkgraaf--Witten theory is determined by the number of
irreducible projective representations of the flux centralizers:
\begin{equation}
 r_\omega(G) = \sum_{[g]} k_{\alpha_g}(C_G(g)),
 \label{eq:rank}
\end{equation}
where $k_\alpha(H)$ counts the irreducible $\alpha$-projective representations of $H$.
By Schur's theory of projective characters \cite{Karpilovsky}, $k_\alpha(H)$ equals
the number of \emph{$\alpha$-regular} conjugacy classes in $H$.
A class represented by $h \in H$ is $\alpha$-regular if $\alpha(h,x) = \alpha(x,h)$
for all $x \in C_H(h)$.

The physical origin of this condition lies in the gauge invariance of torus states
\cite{HW}. Consider a commuting pair $(g,h) \in G^2$, representing flat holonomies
along the two cycles of $T^2$.
Any gauge transformation that commutes with both $g$ and $h$ belongs to the mutual
centralizer $C_G(g,h) = C_G(g) \cap C_G(h)$.
In the projective representation $\pi \in \Irr_{\alpha_g}(C_G(g))$, the operators
associated with $h$ and $x \in C_G(g,h)$ satisfy
\begin{equation}
 \pi(h)\pi(x) = \frac{\alpha_g(h,x)}{\alpha_g(x,h)}\,\pi(x)\pi(h).
\end{equation}
For the state to survive the gauge projection onto physical, gauge-invariant
configurations, this residual commutator phase must equal unity.
Inserting the explicit expression for $\alpha_g$ from \eqref{eq:alpha}, this phase
can be written symmetrically in terms of the commuting triple $(g,h,x)$.
We define the 3-holonomy invariant
\begin{equation}
 \Omega_\omega(g,h,x) :=
 \frac{\omega(g,h,x)\,\omega(h,x,g)\,\omega(x,g,h)}
      {\omega(h,g,x)\,\omega(g,x,h)\,\omega(x,h,g)}.
 \label{eq:omega}
\end{equation}
Physically, $\Omega_\omega(g,h,x)$ is the cocycle weight of the specified flat
connection on $T^3$ with holonomies $(g,h,x)$. The full DW partition function
sums these weights over commuting triples, with the gauge normalization.
The condition that the commuting pair $(g,h)$ is $\alpha_g$-regular is therefore
equivalent to the statement that
\begin{equation}
 \Omega_\omega(g,h,x) = 1 \quad \text{for all } x \in C_G(g,h).
 \label{eq:regular}
\end{equation}
The categorical rank $r_\omega(G)$ counts simultaneous $G$-conjugacy orbits of
commuting pairs $(g,h)$ satisfying \eqref{eq:regular}.
In the untwisted case ($\omega = 1$), $\Omega_1 \equiv 1$ trivially, and every
commuting pair is regular, recovering the classical orbit count
\begin{equation}
 r_1(G) = \sum_{[g]} k(C_G(g)) = \frac{|\Hom(\mathbb{Z}^3,G)|}{|G|}.
 \label{eq:untwisted}
\end{equation}

Crucially, the algebraic properties of the 3-cocycle identity impose powerful
constraints on the phase $\Omega_\omega$.

\begin{lemma}[Alternation of the 3-holonomy phase]\label{lem:alternation}
On any Abelian subgroup $A \le G$, the function $\Omega_\omega: A^3 \to \U$ is an
alternating tricharacter: it is multiplicative in each argument, invariant under
cyclic permutations, and inverts under any odd permutation of arguments.
In particular, $\Omega_\omega(g,h,x) = 1$ whenever any two arguments coincide.
\end{lemma}
\begin{proof}
Multiplicativity follows from expanding the 3-cocycle identity \eqref{eq:background-cocycle}
for pairwise-commuting elements, and coboundary invariance ensures that $\Omega_\omega$
depends only on the cohomology class $[\omega]|_A \in H^3(A,\U)$.
A detailed verification of this identity is provided in Appendix~\ref{app:alternation}.
\end{proof}

This alternation property immediately yields a fundamental physical consequence:
holonomies that reside in the same cyclic subgroup can never be projected out by
any topological twist.

\begin{lemma}[Cyclic protection]\label{lem:protection}
If a commuting pair $(g,h) \in G^2$ generates a cyclic subgroup $\langle g,h\rangle$,
then $(g,h)$ satisfies the regularity condition \eqref{eq:regular} for every
cohomology class $[\omega] \in H^3(G,\U)$.
\end{lemma}
\begin{proof}
Because $\langle g,h\rangle$ is cyclic, there exists a common generator $c \in G$
such that $g = c^i$ and $h = c^j$ for some integers $i,j$.
Any element $x \in C_G(g,h)$ necessarily commutes with $c$, since $c$ can be
expressed as an integer word in $g$ and $h$.
Applying Lemma~\ref{lem:alternation} to the Abelian subgroup $\langle c,x\rangle \le G$,
we obtain
\begin{equation}
 \Omega_\omega(g,h,x) = \Omega_\omega(c^i, c^j, x) = \Omega_\omega(c,c,x)^{ij} = 1^{ij} = 1,
\end{equation}
which establishes universal regularity across all twists $\omega$.
\end{proof}

\subsection{The cyclic protection index}

Lemma~\ref{lem:protection} demonstrates that every simultaneous conjugacy orbit of
commuting pairs $(g,h)$ generating a cyclic subgroup persists across all twists.
Counting these protected orbits yields a universal, twist-independent lower bound
on the categorical rank of $D^\omega(G)$.

Let $C = \langle c\rangle$ be a cyclic subgroup of order $m = |C|$.
A pair $(c^i,c^j) \in C^2$ generates $C$ if and only if $\gcd(i,j,m) = 1$.
The number of such generating pairs is given by the second Jordan totient function,
\begin{equation}
 J_2(m) = m^2 \prod_{\ell \mid m} \left(1 - \frac{1}{\ell^2}\right),
\end{equation}
where the product runs over all distinct prime divisors $\ell$ of $m$.
Two generating pairs in $C^2$ are conjugate in $G$ if and only if they are related
by the normalizer $N_G(C) = \{x \in G : xCx^{-1} = C\}$.
The kernel of this normalizer action is the centralizer $C_G(C)$, and the quotient
$N_G(C)/C_G(C)$ acts faithfully and freely on the set of generating pairs.
Dividing by $[N_G(C) : C_G(C)]$ counts the number of distinct $G$-conjugacy orbits
originating from conjugates of $C$.

\begin{theorem}[Explicit protected-sector bound]\label{thm:cyclic}
For any finite group $G$, define the cyclic protection index
\begin{equation}
 \Xi(G) := \sum_{[C]_{\rm cyc}} \frac{J_2(|C|)}{[N_G(C) : C_G(C)]},
 \label{eq:Xi}
\end{equation}
where the sum runs over all $G$-conjugacy classes of cyclic subgroups $C \le G$.
Then $\Xi(G)$ is an integer, and for every $[\omega] \in H^3(G,\U)$,
\begin{equation}
 \Xi(G) \le r_\omega(G) \le r_1(G).
 \label{eq:sandwich}
\end{equation}
Equivalently, $\Xi(G)$ admits the element-class expansion
\begin{equation}
 \Xi(G) = \sum_{[g]} \ord(g) \prod_{\ell \mid \ord(g)} \left(1 + \frac{1}{\ell}\right),
 \label{eq:element}
\end{equation}
where the sum ranges over all conjugacy classes of elements in $G$.
\end{theorem}
\begin{proof}
For each conjugacy class of cyclic subgroups $[C]$, the generating pairs in $C^2$
are regular for all $\omega$ by Lemma~\ref{lem:protection}.
Because $N_G(C)/C_G(C)$ acts freely on the $J_2(|C|)$ generating pairs, the number
of distinct $G$-orbits of pairs generating a conjugate of $C$ is precisely the
summand in \eqref{eq:Xi}. Summing over all $[C]_{\rm cyc}$ proves that $\Xi(G)$ is
an integer and establishes $\Xi(G) \le r_\omega(G)$.
The upper bound $r_\omega(G) \le r_1(G)$ holds because regular classes form a subset
of all conjugacy classes in each sector.

To derive the element-class formula \eqref{eq:element}, note that Euler's totient
$\varphi(m)$ counts the generators of $C_m$.
The number of distinct $G$-classes of generators of conjugates of $C$ is
$\varphi(|C|)/[N_G(C) : C_G(C)]$.
Using the identity $J_2(m)/\varphi(m) = m \prod_{\ell \mid m} (1 + \ell^{-1})$ and
regrouping the sum in \eqref{eq:Xi} over conjugacy classes of individual elements
$g \in G$ (where $C = \langle g\rangle$) directly yields \eqref{eq:element}.
\end{proof}

The element-class formulation \eqref{eq:element} allows immediate evaluations for
general groups and $p$-groups.

\begin{corollary}\label{cor:classes}
For every finite group $G$ and twist $\omega$,
\begin{equation}
 r_\omega(G) \ge \Xi(G) \ge 3k(G) - 2.
\end{equation}
If $G$ is a $p$-group, then $\Xi(G) \ge (p+1)k(G) - p$, with equality holding
if and only if $G$ has exponent dividing $p$.
\end{corollary}
\begin{proof}
In \eqref{eq:element}, the identity class $[1]$ contributes $\ord(1) = 1$.
Every non-identity class $[g]$ has $\ord(g) \ge 2$, contributing at least
$2(1 + 1/2) = 3$. This gives $\Xi(G) \ge 1 + 3(k(G)-1) = 3k(G)-2$, recovering
and refining the general lower bound observed in Ref.~\cite{HW}.
For a $p$-group, every non-identity element has order $p^k \ge p$, contributing
at least $p(1 + 1/p) = p+1$, with equality if and only if $\ord(g) = p$.
\end{proof}

For a cyclic group $G = C_m$, every pair of elements generates a cyclic subgroup,
giving $\Xi(C_m) = m^2 = r_\omega(C_m)$ for all twists.
For $S_3$, the conjugacy classes of orders 1, 2, and 3 give
$\Xi(S_3) = 1 + 3 + 4 = 8 = r_1(S_3)$, confirming that twisting cannot reduce the
rank of $S_3$.
In general, however, non-cyclic commuting pairs can be projected out by non-trivial
cocycles, allowing $r_\omega(G)$ to drop below $r_1(G)$.

\subsection{Alternating responses and sector counts in Abelian gauge theories}
\label{subsec:response}
\label{sec:elementary}

We now specialize to elementary-Abelian gauge groups $V = \mathbb{F}_p^n$, viewed
as $n$-dimensional vector spaces over the finite field $\mathbb{F}_p$.
Here, the group law is additive, and all group elements commute.
By Lemma~\ref{lem:alternation}, the 3-holonomy invariant $\Omega_\omega: V^3 \to \U$
is an alternating tricharacter.
Writing $\zeta = e^{2\pi i/p}$, there exists a unique alternating trilinear form
$\tau \in \Lambda^3 V^*$ such that
\begin{equation}
 \Omega_\omega(u,v,w) = \zeta^{\tau(u,v,w)}, \qquad u,v,w \in V.
\end{equation}
Conversely, every alternating trilinear form $\tau = \sum_{i<j<k} t_{ijk} e_i^* \wedge e_j^* \wedge e_k^*$
is realized by a 3-cocycle on $V$:
\begin{equation}
 \omega(a,b,c) = \zeta^{\sum_{i<j<k} t_{ijk} a_i b_j c_k}.
 \label{eq:representative}
\end{equation}
The exponent is trilinear, so it satisfies the additive 3-cocycle identity;
its alternation is precisely $\tau$. No division by $6$ is involved, so this
realization also applies in characteristics $2$ and $3$.
The map $[\omega] \mapsto \tau$ is a surjective homomorphism from $H^3(V,\U)$ to
$\Lambda^3 V^*$ with equal-sized fibers.
Minimizing the rank $r_\omega(V)$ over all cocycle twists is therefore equivalent
to minimizing over all alternating trilinear forms $\tau \in \Lambda^3 V^*$.
This reduction concerns rank, not equivalence of the resulting modular categories.

Fixing the magnetic flux $u \in V$, the trilinear form contracts to an alternating
bilinear form $B_u \in \Lambda^2 V^*$, defined by
\begin{equation}
 B_u(v,w) := \tau(u,v,w).
\end{equation}
The condition \eqref{eq:regular} that $(u,v)$ is a regular commuting pair requires
$\tau(u,v,w) = 0$ for all $w \in V$.
This means that $v$ must lie in the radical (kernel) of $B_u$:
\begin{equation}
 \Rad B_u = \{v \in V : B_u(v,w) = 0 \text{ for all } w \in V\}.
\end{equation}
The dimension of $\Rad B_u$ is $n - \rank B_u$, where $\rank B_u$ denotes the
standard matrix rank of the alternating bilinear form.
A two-dimensional subspace $W \le V$ (a projective line in $PG(n-1,p)$) is called
\emph{singular} for $\tau$ if $\tau(W,W,V) = 0$ identically \cite{DSsingular,DSforms}.
This connects the anyon count directly to finite geometry.

\begin{proposition}[Rank and singular-line dictionary]\label{prop:dictionary}
Let $\tau \in \Lambda^3 V^*$ and let $L_\tau$ denote the set of singular projective
lines of $\tau$. Then the categorical rank $r_\tau(V)$ satisfies
\begin{align}
 r_\tau(V) &= \#\{(u,v) \in V^2 : \tau(u,v,-) = 0\}
            = \sum_{u \in V} p^{n - \rank B_u},
 \label{eq:contraction}\\
 r_\tau(V) &= (p+1)p^n - p + (p^2-1)(p^2-p)|L_\tau|.
 \label{eq:lines}
\end{align}
\end{proposition}
\begin{proof}
Because $V$ is Abelian, every pair of elements commutes and conjugacy is trivial.
The regular pairs in flux $u$ are precisely the vectors $v \in \Rad B_u$, yielding
$p^{n - \rank B_u}$ choices, which proves \eqref{eq:contraction}.
Linearly dependent pairs $(u,v)$ satisfy $\tau(u,v,-) = 0$ automatically by
alternation; there are $p^n$ pairs with $u=0$ and $(p^n-1)p$ dependent pairs with
$u \ne 0$, totaling $(p+1)p^n - p$.
Every linearly independent regular pair $(u,v)$ spans a unique singular projective
line $W = \langle u,v\rangle \in L_\tau$.
Since each 2-dimensional subspace possesses $(p^2-1)(p^2-p)$ ordered bases,
summing the dependent and independent contributions gives \eqref{eq:lines}.
\end{proof}

Proposition~\ref{prop:dictionary} demonstrates that minimizing the DW rank
$r_\tau(V)$ is mathematically equivalent to minimizing the number of singular
lines $|L_\tau|$.
However, linear algebra imposes an immediate parity constraint: the matrix rank
of any alternating bilinear form must be an even integer.

\begin{proposition}[Parity obstruction]\label{prop:parity}
For any $n \ge 1$ and any alternating trilinear form $\tau \in \Lambda^3 (\mathbb{F}_p^n)^*$,
\begin{equation}
 r_\tau(\mathbb{F}_p^n) \ge
 \begin{cases}
  p^n + (p^n-1)p, & n \text{ odd},\\
  p^n + (p^n-1)p^2, & n \text{ even}.
 \end{cases}
 \label{eq:parity}
\end{equation}
Equality holds if and only if every non-zero flux $u \ne 0$ achieves the maximum
possible contraction rank: $\rank B_u = n-1$ for odd $n$, and $\rank B_u = n-2$
for even $n$.
\end{proposition}
\begin{proof}
Because $B_u(u,v) = \tau(u,u,v) = 0$, the flux vector $u$ always lies in $\Rad B_u$.
Hence, $\rank B_u \le n-1$. Since $\rank B_u$ must be even, $\rank B_u \le n-1$
when $n$ is odd, and $\rank B_u \le n-2$ when $n$ is even.
The flux $u=0$ contributes $p^n$ to \eqref{eq:contraction}, while each of the $p^n-1$
non-zero fluxes contributes at least $p^{n-(n-1)} = p$ (odd $n$) or $p^{n-(n-2)} = p^2$
(even $n$). Summing these terms yields \eqref{eq:parity}.
\end{proof}

When $n > 3$, geometric constraints prevent the cyclic protection floor from being
attained.

\begin{corollary}[Geometric floor]\label{cor:strict}
For $n > 3$, the singular line set is non-empty ($|L_\tau| \ge 1$), and therefore
\begin{equation}
 r_\tau(\mathbb{F}_p^n) \ge (p+1)p^n - p + (p^2-1)(p^2-p).
\end{equation}
\end{corollary}
\begin{proof}
By the Draisma--Shaw singular-subspace theorem \cite[Theorem 1.1]{DSsingular}, every
alternating trilinear form over a finite field in dimension $n > 3$ admits at least
one singular projective line: $|L_\tau| \ge 1$.
Substituting $|L_\tau| \ge 1$ into \eqref{eq:lines} proves the claim.
\end{proof}

\subsection{First moment and the sharp fixed-prime envelope}
\label{subsec:sharp}

Can one construct twists $\tau$ whose rank matches the linear-in-$|V|$ scale
suggested by the cyclic protection floor?
We resolve this affirmatively by evaluating the exact first moment of $r_\tau(V)$
over the space of alternating trilinear forms.

\begin{theorem}[Exact first moment and optimal scale]\label{thm:mean}
Let $\tau$ be chosen uniformly at random from $\Lambda^3(\mathbb{F}_p^n)^*$ for $n \ge 3$.
Then the expected categorical rank is
\begin{align}
 \mathbb{E}[r_\tau(\mathbb{F}_p^n)]
 &= (p+1)p^n - p + p^{-(n-2)}(p^n-1)(p^n-p)
 \label{eq:mean}\\
 &= (p^2+p+1)p^n - (p^3+p^2+p) + p^{3-n}. \notag
\end{align}
Consequently, the minimal rank $m_p(n) := \min_{\omega} r_\omega(\mathbb{F}_p^n)$ satisfies
\begin{equation}
 (p+1)p^n - p \le m_p(n) < (p^2+p+1)p^n,
 \label{eq:bounds}
\end{equation}
establishing $m_p(n) = \Theta_p(p^n)$ as $n \to \infty$.
\end{theorem}
\begin{proof}
Linearly dependent pairs $(u,v)$ satisfy $\tau(u,v,-) = 0$ identically for every $\tau$,
contributing $(p+1)p^n - p$ to the expectation.
For any fixed linearly independent pair $(u,v)$, the linear evaluation map
$\tau \mapsto \tau(u,v,-)$ maps $\Lambda^3 V^*$ surjectively onto the annihilator
$\operatorname{Ann}(\langle u,v\rangle) \subset V^*$.
Because $\dim\operatorname{Ann}(\langle u,v\rangle) = n-2$, the condition $\tau(u,v,-) = 0$
imposes $n-2$ independent linear constraints on the coefficients of $\tau$.
For a uniformly distributed form, the probability that an independent pair vanishes
is precisely $p^{-(n-2)}$.
There are $(p^n-1)(p^n-p)$ independent pairs. By linearity of expectation,
\begin{equation}
 \mathbb{E}[r_\tau(\mathbb{F}_p^n)] = (p+1)p^n - p + p^{-(n-2)}(p^n-1)(p^n-p),
\end{equation}
which simplifies directly to \eqref{eq:mean}.
Since the minimum cannot exceed the average, $m_p(n) \le \mathbb{E}[r_\tau] < (p^2+p+1)p^n$.
The lower bound is guaranteed by cyclic protection (Proposition~\ref{prop:dictionary}).
\end{proof}

While Theorem~\ref{thm:mean} establishes existence via the probabilistic method,
Markov's inequality supplies a constant-factor sampling guarantee: for every
$t>1$, a uniform form has $r_\tau \le t(p^2+p+1)p^n$ with probability at least
$1-1/t$. This does not assert tight concentration or an optimal leading constant.
Inverting this relation under the rank cutoff $r \le R$ yields our main sharp
envelope theorem for elementary-Abelian theories.

\begin{corollary}[Sharp fixed-prime envelopes]\label{cor:envelope}
For a fixed prime $p$, the twisted and untwisted envelopes satisfy
\begin{equation}
 F_p(R) = 2\log R + O_p(1), \qquad
 F_p^{(1)}(R) = \log R + O_p(1),
 \label{eq:envelopes}
\end{equation}
as $R \to \infty$. In terms of maximal total quantum dimension,
\begin{equation}
 \cD_{\max,p}(R) = \Theta_p(R), \qquad
 \cD_{\max,p}^{(1)}(R) = \Theta_p(R^{1/2}).
\end{equation}
Explicitly, setting $A_p = p^2+p+1$, for all $R \ge A_p p^3$ we have
\begin{equation}
 2\log R - 2\log(pA_p) \le F_p(R) \le 2\log\left(\frac{R+p}{p+1}\right),
\end{equation}
while $F_p^{(1)}(R) = 2\log p \, \lfloor\log R / (2\log p)\rfloor$ for $R \ge 1$.
\end{corollary}
\begin{proof}
The lower bound $r_\omega(\mathbb{F}_p^n) \ge (p+1)p^n - p$ holds for all $n \ge 0$,
implying $p^n \le (R+p)/(p+1)$ whenever $r_\omega \le R$, which establishes the upper bound
$F_p(R) \le 2\log((R+p)/(p+1))$.
Conversely, choose $n$ such that $p^n$ is the largest power of $p$ satisfying
$p^n \le R/A_p$. For $R \ge A_p p^3$, we have $n \ge 3$ and $p^n > R/(p A_p)$.
By Theorem~\ref{thm:mean}, there exists a twist with $r_\tau(\mathbb{F}_p^n) < A_p p^n \le R$,
rendering this group order admissible. Thus, $F_p(R) \ge 2n\log p > 2\log(R/(pA_p))$.
In the untwisted case, $r_1(\mathbb{F}_p^n) = p^{2n}$, giving $p^{2n} \le R$ and
yielding the exact floor formula.
\end{proof}

Expressed in terms of the binary logarithmic cutoff $q_R = \log_2 R$, the envelopes read
\begin{equation}
 F_p(R) = 2(\log 2)q_R + O_p(1), \qquad
 F_p^{(1)}(R) = (\log 2)q_R + O_p(1).
 \label{eq:binary-laws}
\end{equation}
Introducing topological twists thus doubles the leading coefficient of the envelope,
elevating the total quantum dimension from $\Theta_p(R^{1/2})$ to $\Theta_p(R)$.

\subsection{Finite spectra and block-decomposition limits}
\label{subsec:examples}

Small dimensions distinguish exact attainment of the rank lower bounds from
optimality only at the level of asymptotic scaling. We begin with the volume
form in dimension 3, compare the intervening binary dimensions 4 and 5, and
then give a construction attaining the parity bound in dimension 6.
For $n=3$, any non-zero alternating trilinear form $\tau$ is a volume form.
For every non-zero flux $u \ne 0$, $\rank B_u = 2$ and $\Rad B_u = \langle u\rangle$.
Thus, $m_p(3) = (p+1)p^3 - p$ saturates the cyclic protection bound.
The resulting theory hosts $p^3$ Abelian anyons ($d_a = 1$) from the identity flux,
and $(p^3-1)p$ non-Abelian anyons ($d_a = p$) from the $p^3-1$ non-zero fluxes,
yielding $\sum_a d_a^2 = p^3(1)^2 + (p^3-1)p(p)^2 = p^6$, exactly verifying $\cD = p^3$.

Attainment does not persist in the next two dimensions. The exhaustive binary
census in Appendix~\ref{app:census} gives
\[
 m_2(4)=88>76=2^4+(2^4-1)2^2,\qquad
 m_2(5)=184>94=2^5+(2^5-1)2.
\]
Thus, even an optimal twist need not attain the parity lower bound: no form
in either of these two binary dimensions has maximal contraction rank at
every non-zero flux. These cases contrast with the uniform non-zero-flux
spectra in dimensions 3 and 6.

In dimension $n=6$, the parity obstruction (Proposition~\ref{prop:parity}) dictates
$m_p(6) \ge p^6 + (p^6-1)p^2$. This bound is saturated by an algebraic
construction based on quadratic field extensions.

\begin{proposition}[Optimality of the field-trace line spread]\label{prop:six}
For every prime $p$, the minimal rank in dimension $n=6$ is
\begin{equation}
 m_p(6) = p^6 + (p^6-1)p^2.
\end{equation}
In particular, for binary gauge fields ($p=2$), $m_2(6) = 316$.
\end{proposition}
\begin{proof}
Let $K = \mathbb{F}_{p^2}$ be the quadratic extension of $\mathbb{F}_p$, so that
$V = K^3$ is a 6-dimensional vector space over $\mathbb{F}_p$.
Following the line-spread construction of Draisma and Shaw \cite{DSforms}, define
\begin{equation}
 \tau(u,v,w) = \Tr_{K/\mathbb{F}_p}[\det_K(u,v,w)],
\end{equation}
where $\Tr_{K/\mathbb{F}_p}(z) = z + z^p$ is the field trace.
For any non-zero flux $u \in V \setminus \{0\}$, the 2-dimensional $\mathbb{F}_p$-subspace
$Ku$ satisfies $\tau(u, Ku, V) = 0$, so $Ku \subseteq \Rad B_u$.
If $v \notin Ku$, then $u$ and $v$ are $K$-linearly independent, and the map
$w \mapsto \det_K(u,v,w)$ is a surjective $K$-linear map onto $K$.
Because the trace $\Tr_{K/\mathbb{F}_p}: K \to \mathbb{F}_p$ is non-zero, $\tau(u,v,-)$
cannot vanish identically.
Hence, $\Rad B_u = Ku$ has dimension exactly 2 over $\mathbb{F}_p$, giving
$\rank B_u = 4$ for every $u \ne 0$.
By \eqref{eq:contraction}, $r_\tau(\mathbb{F}_p^6) = p^6 + (p^6-1)p^2$, which saturates
the even-dimensional parity lower bound \eqref{eq:parity}.
\end{proof}

The singular lines of this optimal form are precisely the 1-dimensional $K$-subspaces
viewed as 2-dimensional $\mathbb{F}_p$-subspaces, forming a Desarguesian line spread
that partitions $V \setminus \{0\}$.
The spectrum consists of $p^6$ Abelian anyons ($d_a = 1$) and $(p^6-1)p^2$ non-Abelian
anyons of dimension $d_a = p^2$, satisfying $\sum_a d_a^2 = p^{12} = \cD^2$.

It is instructive to note why simple block-diagonal extensions of these exceptional
small-dimensional forms fail to produce the asymptotic linear scale.
If $V = V_1 \oplus V_2$ and $\tau = \tau_1 \oplus \tau_2$ has no cross-terms, the
contraction radicals decouple: $\Rad B_{(u_1,u_2)} = \Rad B_{u_1} \oplus \Rad B_{u_2}$.
Consequently, the rank factorizes multiplicatively:
\begin{equation}
 r_\tau(V_1 \oplus V_2) = r_{\tau_1}(V_1)\,r_{\tau_2}(V_2).
 \label{eq:product}
\end{equation}
Iterating a fixed block $V_0$ of dimension $n_0$ yields
$r(V_0^{\oplus t}) = [r(V_0)]^t = |V_0^{\oplus t}|^{\log r(V_0)/\log|V_0|}$.
Because $r(V_0) > |V_0|$ for any non-trivial block, the scaling exponent is strictly
greater than 1 (for instance, the $n=3$ block gives exponent $\log_p((p+1)p^3-p)/3 > 1$).
Similarly, defining $\tau$ on $\mathbb{F}_{p^s}^3$ via $\Tr_{\mathbb{F}_{p^s}/\mathbb{F}_p}\det$
yields rank $p^{3s} + (p^{3s}-1)p^s = \Theta(|V|^{4/3})$ as $s \to \infty$.
These fixed-block and field-trace families do not realize the linear scale.
Theorem~\ref{thm:mean} establishes the existence of forms that do, without
determining their decomposition type. Indeed, a fixed number of growing
blocks with $r_{\tau_i}(V_i)=O_p(|V_i|)$ also gives linear-scale rank by
\eqref{eq:product}.

\section{Gauge-Group Structure and Untwisted Envelopes}
\label{sec:structure}
\label{sec:constraints}

Having explored how topological cocycle twists compress the anyon spectrum at
fixed gauge group, we now investigate the dual mechanism: varying the algebraic
architecture of the gauge group $G$.
The central physical question is whether one can design non-Abelian gauge groups
whose total quantum dimension $\cD = |G|$ grows rapidly without generating an
excessive proliferation of magnetic flux sectors and electric charges.
Except where twists are explicitly considered in the solvable, radical-free, and
global bounds, this section focuses on the untwisted setting ($\omega = 1$).

We first establish sharp asymptotic envelopes for permutation gauge groups ($S_n$
and $A_n$) and their wreath products.
We then prove an absolute quadratic-logarithmic envelope bound for almost-simple
groups and their direct products.
Next, we examine the structural costs of group extensions, proving that taking
quotients can increase categorical rank and establishing the exact sharpness of
this quotient loss using extraspecial groups.
Finally, we derive structural refinements across solvable and radical-free groups,
and conclude with the global finiteness theorem.
Throughout this analysis, the following standard counting relations serve as our
basic tools.

\begin{lemma}[Counting identities on commuting triples]\label{lem:basic}
For every finite group $G$, the untwisted categorical rank is given by
\begin{equation}
 r_1(G) = \sum_{[x] \subset G} k(C_G(x)).
 \label{eq:centralizers}
\end{equation}
Moreover, $r_1$ is strictly multiplicative under direct products,
$r_1(G \times H) = r_1(G)\,r_1(H)$, and for any subgroup $H \le G$,
\begin{equation}
 r_1(H) \le [G : H]\,r_1(G).
 \label{eq:subgroup}
\end{equation}
\end{lemma}
\begin{proof}
By Burnside's lemma, the count of ordered commuting triples $h_3(G) = |\Hom(\mathbb{Z}^3,G)|$
can be evaluated by fixing the first element $x \in G$:
\begin{equation}
 h_3(G) = \sum_{x \in G} |C_G(x)|\,k(C_G(x)),
\end{equation}
since any finite group $K$ has $|K|k(K)$ commuting pairs.
Partitioning the sum into conjugacy classes $[x]$ yields
\begin{equation}
 h_3(G) = \sum_{[x]} |[x]|\,|C_G(x)|\,k(C_G(x)) = |G|\sum_{[x]} k(C_G(x)),
\end{equation}
which proves \eqref{eq:centralizers} upon dividing by $|G|$.
Direct-product multiplicativity follows because $\Hom(\mathbb{Z}^3, G \times H)$
naturally factors as $\Hom(\mathbb{Z}^3,G) \times \Hom(\mathbb{Z}^3,H)$.
Finally, any commuting triple in $H$ is a commuting triple in $G$, so $h_3(H) \le h_3(G)$;
dividing by $|H| = |G|/[G:H]$ establishes \eqref{eq:subgroup}.
\end{proof}

\subsection{Permutation gauge groups and sharp restricted envelopes}
\label{sec:permutation}

\subsubsection{Symmetric and alternating gauge groups}
Permutation gauge groups—the symmetric group $S_n$ and alternating group $A_n$—provide
a natural arena for testing how non-Abelian statistics dilute the categorical rank.
We denote their untwisted envelopes by $F_S^{(1)}(R)$ and $F_A^{(1)}(R)$,
restricting the gauge group to $S_n$ and $A_n$, respectively.
In an untwisted gauge theory with $G = S_n$, flat connections on the spatial 2-torus
correspond to commuting pairs of permutations $(\sigma, \tau) \in S_n^2$, modulo
simultaneous conjugation.
Under simultaneous conjugation, the commuting pair $(\sigma,\tau)$ defines an action
of the fundamental group $\mathbb{Z}^2$ on the set of $n$ points.
The set of physical torus states therefore corresponds bijectively to isomorphism
classes of finite $\mathbb{Z}^2$-sets of cardinality $n$.

Every finite $\mathbb{Z}^2$-set decomposes uniquely into a disjoint union of transitive
orbits. A transitive $\mathbb{Z}^2$-orbit of size $m$ is isomorphic to the quotient
$\mathbb{Z}^2/L$, where $L \le \mathbb{Z}^2$ is a sublattice of index $[\mathbb{Z}^2 : L] = m$.
The number of such sublattices of index $m$ is given by the sum-of-divisors function
$\sigma_1(m) = \sum_{d \mid m} d$.
Because any transitive orbit of size $m$ can appear with arbitrary non-negative
integer multiplicity, the generating function for the sequence of ranks $a_n := r_1(S_n)$
takes an exact Euler product form \cite{BF,BFH}:
\begin{equation}
 \sum_{n \ge 0} a_n t^n = \prod_{m \ge 1} (1 - t^m)^{-\sigma_1(m)}.
 \label{eq:symmetric}
\end{equation}
Using the analytic properties of the Dirichlet series $\sum_{m \ge 1} \sigma_1(m)m^{-s} = \zeta(s)\zeta(s-1)$,
Bringmann, Franke, and Heim \cite{BFH} established the precise asymptotic growth of the
coefficients $a_n$:
\begin{equation}
 \log a_n \sim B n^{2/3}, \qquad
 B = \frac{3}{2}(2Z)^{1/3}, \qquad
 Z = \zeta(2)\zeta(3).
\end{equation}
To determine the physical envelope, we invert this growth to find the largest degree
$n$ permitted by a rank cutoff $R$, and evaluate the total quantum dimension
$\cD(S_n,1) = n!$.

\begin{corollary}[Sharp permutation envelopes]\label{cor:symmetric}
Under the categorical rank cutoff $r_1 \le R$, the untwisted envelopes for symmetric
and alternating gauge groups satisfy
\begin{equation}
 F_S^{(1)}(R) \sim F_A^{(1)}(R) \sim \frac{2}{\sqrt{3\zeta(2)\zeta(3)}}\,(\log R)^{3/2}\log\log R,
 \label{eq:symenv}
\end{equation}
as $R \to \infty$.
\end{corollary}
\begin{proof}
For any $\varepsilon > 0$ and all sufficiently large $n$, we have
$(B - \varepsilon)n^{2/3} \le \log a_n \le (B + \varepsilon)n^{2/3}$.
Imposing the cutoff $a_n \le R$, the maximal admissible degree is
$n(R) = (B^{-1}\log R)^{3/2}(1 + o(1))$.
By Stirling's approximation, the total quantum dimension satisfies
$2\log(n!) \sim 2n\log n$.
Substituting $n(R)$ yields
\begin{equation}
 2\log(n!) \sim 2\left(\frac{\log R}{B}\right)^{3/2} \log\left[\left(\frac{\log R}{B}\right)^{3/2}\right]
 \sim 3B^{-3/2}(\log R)^{3/2}\log\log R.
\end{equation}
Noting that $3B^{-3/2} = 3 \cdot [ \frac{3}{2}(2Z)^{1/3} ]^{-3/2} = 2/\sqrt{3Z}$ proves the
asymptotic for $S_n$.

For the alternating group $A_n$, the embedding $A_n \le S_n$ with $[S_n : A_n] = 2$ and
the embedding $S_{n-2} \hookrightarrow A_n$ (defined by $\sigma \mapsto \sigma (n-1,n)^{\epsilon(\sigma)}$,
where $\epsilon(\sigma) \in \{0,1\}$ is the permutation parity, so the sign is
$(-1)^{\epsilon(\sigma)}$) combine with Lemma~\ref{lem:basic}
to give
\begin{equation}
 \frac{r_1(S_{n-2})}{n(n-1)/2} \le r_1(A_n) \le 2r_1(S_n).
\end{equation}
Taking logarithms reveals $\log r_1(A_n) \sim B n^{2/3}$.
Because $\log|A_n| = \log(n!) - \log 2 \sim n\log n$, inverting the rank constraint yields
the identical leading asymptotic \eqref{eq:symenv}.
\end{proof}

The physical origin of the scale $(\log R)^{3/2}\log\log R$ lies in the pronounced
imbalance between the factorial growth of the gauge group order and the sub-exponential
growth of the torus ground-state degeneracy:
\begin{equation}
 \log\cD(S_n,1) \sim n\log n, \qquad
 \log r_1(S_n) \sim \frac{3}{2}[2\zeta(2)\zeta(3)]^{1/3} n^{2/3}.
 \label{eq:physical-symmetric-scale}
\end{equation}
The ratio $F_S^{(1)}(R)/(\log_2 R)^2$ therefore tends to zero,
with decay $O(n^{-1/3}\log n)$.
Consequently, while permutation theories achieve substantially higher total quantum
dimensions than elementary-Abelian theories at the same rank cutoff, they remain
strictly sub-quadratic on the logarithmic scale.

\subsubsection{Fixed-base wreath products}
To test whether adding a fixed internal gauge group at each permutation site
improves the permutation envelope, decorate each site with degrees of freedom
from a fixed finite group $H$, forming the wreath product
\begin{equation}
 H \wr S_n = H^n \rtimes S_n, \qquad |H \wr S_n| = |H|^n n!.
\end{equation}
Here, $S_n$ acts by permuting the $n$ factors of $H$.
In the commuting-pair picture, each transitive $\mathbb{Z}^2$-orbit of size $m$
is decorated by an $H$-local system—an equivalence class of homomorphisms from the
sublattice $L \cong \mathbb{Z}^2$ into $H$, modulo conjugation by $H$.
Because $L \cong \mathbb{Z}^2$, the number of distinct $H$-local systems on an
orbit is precisely $r_1(H)$. This is the local-system enumeration underlying
the wreath-product identities of Ref.~\cite{FS}; we use it here to determine
the rank-cutoff envelope.

\begin{proposition}[Wreath product envelope]\label{prop:symwreath}
Fix a finite group $H$ and let $r = r_1(H)$. Then the generating function for
$H \wr S_n$ is
\begin{equation}
 \sum_{n \ge 0} r_1(H \wr S_n) t^n = \prod_{m \ge 1} (1 - t^m)^{-r \sigma_1(m)},
 \label{eq:wreathproduct}
\end{equation}
and the rank grows asymptotically as
\begin{equation}
 \log r_1(H \wr S_n) \sim B r^{1/3} n^{2/3}.
 \label{eq:wreathasymptotic}
\end{equation}
Consequently, the envelope for $H \wr S_n$ obeys the same functional form as $S_n$,
with a strictly suppressed leading coefficient when $H\ne1$:
\begin{equation}
 F_{H\wr S}^{(1)}(R) \sim \frac{2}{\sqrt{3r_1(H)\zeta(2)\zeta(3)}}\,(\log R)^{3/2}\log\log R.
\end{equation}
\end{proposition}
\begin{proof}
Because each transitive orbit admits $r = r_1(H)$ distinct local system decorations,
the generating function is the $r$-th power of the symmetric group generating function
\eqref{eq:symmetric}, proving \eqref{eq:wreathproduct}.
The coefficient of $t^n$ is given by the convolution $\sum_{n_1+\cdots+n_r=n} \prod_{i=1}^r a_{n_i}$.
For the lower bound, setting $n_i = n/r + O(1)$ gives
$\sum_i \log a_{n_i} \sim r B (n/r)^{2/3} = B r^{1/3} n^{2/3}$.
For the upper bound, for every $\epsilon>0$ choose $C_\epsilon$ such that
$\log a_j\le(B+\epsilon)j^{2/3}+C_\epsilon$ for all $j\ge0$.
Concavity gives $\sum_{i=1}^r n_i^{2/3} \le r^{1/3} n^{2/3}$, uniformly
including small or zero parts.
Because the number of compositions of $n$ into $r$ parts is $\binom{n+r-1}{r-1} \le (n+1)^{r-1}$,
and $r$ is fixed, $\log\binom{n+r-1}{r-1} = O(\log n) = o(n^{2/3})$.
Thus the logarithm of the convolution is at most
$(B+\epsilon)r^{1/3}n^{2/3}+rC_\epsilon+O_r(\log n)$.
Letting $\epsilon\downarrow0$ establishes \eqref{eq:wreathasymptotic}.
Because $2\log|H \wr S_n| = 2n\log n + O_H(n)$, inverting the rank constraint
yields the leading coefficient $3(B r^{1/3})^{-3/2} = 2/\sqrt{3r Z}$.
\end{proof}
Every non-trivial group has $r_1(H) \ge 4$, so decorating permutation sites with a
fixed base group strictly reduces the efficiency of the envelope.
This asymptotic argument holds for fixed $H$; it gives no uniform estimate
when the base group varies with $n$.

\subsection{Almost-simple gauge groups and their direct products}
\label{subsec:almostsimple}

We now examine non-Abelian simple groups and their extensions.
A finite group $S$ is simple if its only normal subgroups are $1$ and $S$.
An \emph{almost-simple} group is a group $G$ satisfying
\begin{equation}
 S \le G \le \operatorname{Aut}(S),
\end{equation}
where $S$ is a non-Abelian simple group (the socle of $G$).
The quotient $\operatorname{Out}(S) = \operatorname{Aut}(S)/S$ is the outer automorphism
group.
By the classification of finite simple groups, non-Abelian simple groups fall into
groups of Lie type, alternating groups $A_n$ ($n \ge 5$), and 26 sporadic groups.
For a specified Lie-type family, including its twisting data, the remaining
parameters are the Lie rank $\ell$ and defining field size $q$.

We prove that almost-simple gauge groups and their arbitrary finite direct products
satisfy a quadratic-logarithmic upper bound.

\begin{theorem}[Almost-simple quadratic-logarithmic bound]\label{thm:almostsimple}
There exists an absolute constant $C > 0$ such that
\begin{equation}
 2\log|G| \le C(\log r_1(G))^2
 \label{eq:almostsimple}
\end{equation}
for every finite group $G$ that is almost simple, or a finite direct product of
almost-simple groups.
Consequently, the restricted untwisted envelope satisfies
\begin{equation}
 F_{\mathrm{a.s.}}^{(1)}(R) \le C(\log R)^2.
\end{equation}
The same bound holds if cyclic groups of prime order are included as direct factors.
\end{theorem}
\begin{proof}
We divide the proof according to the classification of simple socles $S$.

\emph{Lie-type socle.}
Let $S$ be a simple group of Lie type of rank $\ell$ over a finite field of size $q$.
By standard bounds on Lie-type orders and outer automorphisms \cite{GLS}, along with
the class-number estimates of Fulman and Guralnick
\cite[Theorem 1.1 and Sections 3--5]{FG}, there exist
absolute positive constants $C_0, c_0, C_1, C_2, C_3$ such that
\begin{align}
 \log|S| &\le C_0 \ell^2 \log q, \notag\\
 \log k(S) &\ge c_0 \ell \log q - C_1 \log(\ell+1) - C_2, \notag\\
 |\operatorname{Out}(S)| &\le C_3 (\ell+1)\log(2q).
 \label{eq:lie-input}
\end{align}
For the classical families, the algebraic-group class counts pass to simple
central quotients with a loss at most polynomial in $\ell+1$: a central
quotient can merge at most the order of its center many classes, and the
additional indices are bounded. The exceptional families have bounded rank.
Their class counts, including those of Suzuki and Ree groups, grow at least
as a fixed positive power of the field size. Reducing $c_0$ therefore gives
the uniform coarse bound in \eqref{eq:lie-input}, without requiring a literal
$q^\ell$ estimate in every field-size convention.
Letting $x = \ell \log q \ge \log 2$, all subleading terms in \eqref{eq:lie-input}
are bounded by $O(\log(2+x))$.
For any intermediate group $S \le G \le \operatorname{Aut}(S)$, Lemma~\ref{lem:basic} gives
\begin{equation}
 \log r_1(G) \ge \log r_1(S) - \log[G : S] \ge \log k(S) - \log|\operatorname{Out}(S)|
 \ge c_0 x - O(\log(2+x)).
\end{equation}
For sufficiently large $x$, this yields $\log r_1(G) \ge c x$ for some absolute $c > 0$.
Conversely, the group order satisfies
\begin{equation}
 \log|G| \le \log|S| + \log|\operatorname{Out}(S)| \le C_0 \ell^2 \log q + O(\log(2+x))
 \le C' x^2,
\end{equation}
since $\ell \ge 1$ and $\log q \ge \log 2$.
Combining these inequalities gives $2\log|G| \le C(\log r_1(G))^2$.
The finitely many groups with bounded $x$ are accommodated by increasing $C$.

\emph{Alternating and sporadic socles.}
For an alternating socle $S = A_n$, Corollary~\ref{cor:symmetric} gives $\log r_1(A_n) \sim B n^{2/3}$.
Because $|\operatorname{Out}(A_n)| \le 4$ for all $n$, we have $\log r_1(G) \ge c n^{2/3}$
for large $n$, while $\log|G| \sim n\log n = O(n^{4/3})$.
Thus, $2\log|G| = O((\log r_1(G))^{3/2}\log\log r_1(G)) = o((\log r_1(G))^2)$.
The 26 sporadic groups and their automorphism groups form a finite set, absorbed
by adjusting $C$.

\emph{Direct products.}
Let $G = \prod_{i=1}^m G_i$, where each factor $G_i$ is almost simple.
Because $r_1$ and $|G|$ are strictly multiplicative (Lemma~\ref{lem:basic}),
\begin{equation}
 2\log|G| = \sum_{i=1}^m 2\log|G_i|
 \le C \sum_{i=1}^m (\log r_1(G_i))^2
 \le C \left(\sum_{i=1}^m \log r_1(G_i)\right)^2
 = C (\log r_1(G))^2.
\end{equation}
If cyclic prime-order factors $C_p$ are included, $r_1(C_p) = p^2$, which satisfies
$2\log p \le (2\log 2)^{-1}(\log r_1(C_p))^2$.
Enlarging $C$ to exceed $(2\log 2)^{-1}$ covers these factors simultaneously.
\end{proof}

Theorem~\ref{thm:almostsimple} proves the upper-bound direction of the
quadratic-logarithmic conjecture for almost-simple groups and their direct
products: $\cD_{\max,\mathrm{a.s.}}^{(1)}(R) \le \exp(C(\log R)^2/2)$.
It does not produce a family saturating this scale, nor does it cover all
radical-free groups.

\subsection{Rank cost of gauge-group quotients}
\label{subsec:quotients}

In building larger gauge groups from simpler constituents, group extensions play a
fundamental role.
Let $N \triangleleft G$ be a normal subgroup with $|N| = a$, and consider the quotient
group $Q = G/N$.
The total quantum dimension changes monotonically: $\cD(G,1) = a \cD(Q,1)$.
However, the categorical rank $r_1$ does not obey simple monotonicity: commuting pairs
in the quotient $Q$ do not necessarily lift to commuting pairs in $G$.
Consequently, passing to a quotient can actually \emph{increase} the rank per unit
quantum dimension.

These are comparisons of gauge-group presentations; no particular anyon
condensation process or microscopic interpolation is assumed.

The following lemma bounds the behavior of conjugacy class numbers under subgroups
and quotients.

\begin{lemma}[Class-number relations]\label{lem:classes}
If $H \le K$ and $M \triangleleft K$, then
\begin{equation}
 k(K) \le [K : H]\,k(H), \qquad k(K/M) \le k(K).
\end{equation}
Furthermore, if $M \le Z(K)$ is central with $|M| = a$, then
\begin{equation}
 k(K) \ge k(K/M) + a - 1.
 \label{eq:surplus}
\end{equation}
\end{lemma}
\begin{proof}
For the subgroup inequality, each irreducible character $\chi \in \Irr(K)$ restricts
to $H$ as a sum of constituents $\theta \in \Irr(H)$.
By Frobenius reciprocity, $\chi$ occurs in the induced character $\operatorname{Ind}_H^K \theta$.
Every such constituent has degree $\chi(1)\ge\theta(1)$, since its restriction
contains $\theta$. Because $\dim(\operatorname{Ind}_H^K \theta) = [K:H]\dim\theta$, at most $[K:H]$ distinct
irreducible characters of $K$ can contain $\theta$ as a constituent, proving
$k(K) \le [K:H]k(H)$ \cite{Gallagher}.
For quotients, inflation embeds $\Irr(K/M)$ injectively into $\Irr(K)$, so $k(K/M) \le k(K)$.
When $M \le Z(K)$, the identity class of $K/M$ lifts to $a$ distinct central singleton
classes in $K$, while every other class in $K/M$ has at least one pre-image class in $K$,
yielding \eqref{eq:surplus}.
\end{proof}

Using this lifting analysis, we bound the rank inflation under arbitrary group quotients.

\begin{theorem}[Quotient loss bound]\label{thm:quotient}
Let $N \triangleleft G$ with $a = |N|$, and let $Q = G/N$. Then
\begin{equation}
 r_1(Q) \le a\,r_1(G), \qquad h_3(Q) \le h_3(G).
 \label{eq:quotient}
\end{equation}
Consequently, $h_3(S) \le h_3(G)$ for every section $S = H/K$ of $G$ (where $K \triangleleft H \le G$).
\end{theorem}
\begin{proof}
Let $\pi: G \to Q$ be the canonical projection.
For each conjugacy class $[\bar{x}] \subset Q$, choose a representative $\bar{x} \in Q$
and a lift $x \in G$ with $\pi(x) = \bar{x}$.
Define the pre-image subgroup $X_x := \pi^{-1}(C_Q(\bar{x}))$ and the projected centralizer
$H_x := \pi(C_G(x))$.
The $X_x$-conjugates of $x$ lie entirely within the coset $xN$, so $[X_x : C_G(x)] \le |N| = a$.
Using the isomorphism theorem $[C_G(x)N : C_G(x)] = [N : C_N(x)]$, we find
\begin{equation}
 [C_Q(\bar{x}) : H_x] = [X_x : C_G(x)N] \le |C_N(x)| \le a.
 \label{eq:index}
\end{equation}
Because $H_x \cong C_G(x)/C_N(x)$, Lemma~\ref{lem:classes} yields
\begin{equation}
 k(C_Q(\bar{x})) \le |C_N(x)|\,k(H_x) \le |C_N(x)|\,k(C_G(x)).
\end{equation}
Because the chosen lifts $x$ of distinct quotient conjugacy classes belong to distinct
$G$-conjugacy classes, summing over all $[\bar{x}] \subset Q$ gives
\begin{equation}
 r_1(Q) = \sum_{[\bar{x}] \subset Q} k(C_Q(\bar{x}))
 \le \sum_{[\bar{x}] \subset Q} |C_N(x)|\,k(C_G(x))
 \le a \sum_{[x] \subset G} k(C_G(x)) = a\,r_1(G).
\end{equation}
Multiplying by $|Q| = |G|/a$ yields $h_3(Q) \le h_3(G)$.
Composing this quotient monotonicity with subgroup monotonicity $h_3(H) \le h_3(G)$
proves the assertion for arbitrary sections.
\end{proof}

\begin{corollary}[Commuting probability comparison]\label{cor:prob}
For any normal subgroup $N \triangleleft G$ of order $a = |N|$, the commuting probabilities
satisfy
\begin{equation}
 P_3(G) \le P_3(G/N) \le a^3 P_3(G).
\end{equation}
\end{corollary}
\begin{proof}
The lower bound $P_3(G) \le P_3(G/N)$ reflects the fact that commuting in $G$ implies
commuting in the quotient $G/N$ \cite{LSprob}.
The upper bound follows directly from $h_3(G/N) \le h_3(G)$ via the normalization
$P_3(G) = h_3(G)/|G|^3$, noting $|G|^3 = a^3 |G/N|^3$.
\end{proof}

\subsection{Central extensions and flux-resolved constraints}
\label{subsec:central}
\label{subsec:sharp-loss}

The center $Z(G)$ consists of elements commuting with every element
of $G$. An extension is central when its kernel lies in $Z(G)$.
In this case, the obstruction to lifting commuting elements can be resolved
flux by flux.

Assume $A \le Z(G)$, $a = |A|$, $Q = G/A$, and choose a lift $x \in G$ of $\bar{x} \in Q$.
With commutator convention $[x,y] = x^{-1}y^{-1}xy$, define
\begin{equation}\label{eq:phi}
 \phi_x: C_Q(\bar{x}) \longrightarrow A, \qquad
 \phi_x(\bar{y}) = [x,y], \qquad
 B_x = \operatorname{im}\phi_x, \quad b_x = |B_x|.
\end{equation}
Because $A$ is central, the commutator $[x,y] \in A$ is independent of the choice
of lifts, and $\phi_x$ is a well-defined group homomorphism.
The integer $b_x$ measures the failure of the centralizer $C_Q(\bar{x})$ to lift
to commuting elements with $x$. When $b_x = 1$, every element in $C_Q(\bar{x})$ lifts
to an element commuting with $x$. In general, $a/b_x$ counts how many flux classes
lie above $[\bar{x}]$.

\begin{theorem}[Central lifting and weighted correction]\label{thm:central}
In this notation,
\begin{align}
 k(G) &= \sum_{[\bar{x}] \subset Q} \frac{a}{b_x}, \label{eq:classlift}\\
 r_1(G) &= \sum_{[\bar{x}] \subset Q} \frac{a}{b_x} k(C_G(x)). \label{eq:exact}
\end{align}
Furthermore,
\begin{equation}\label{eq:weighted}
 a\,r_1(G) \ge \sum_{[\bar{x}] \subset Q} \left(\frac{a}{b_x}\right)^2 k(C_Q(\bar{x})) + a(a-1)k(G).
\end{equation}
In particular,
\begin{equation}\label{eq:centralbound}
 a\,r_1(G) \ge r_1(Q) + (a^2-1)k(Q) + a(a-1)k(G).
\end{equation}
Dropping the last term gives $a\,r_1(G) \ge r_1(Q) + (a^2-1)k(Q)$.
\end{theorem}
\begin{proof}
Let $K_x = \ker\phi_x$. Then $C_G(x)/A \cong K_x$, and $|C_G(x)| = a|C_Q(\bar{x})|/b_x$.
There are $a[Q : C_Q(\bar{x})]$ elements in the coset pre-image of $[\bar{x}]$.
Each $G$-class above it has size $[Q : C_Q(\bar{x})]b_x$, because conjugating a
representative image to $\bar{x}$ gives a lift $xz$ ($z \in A$) with $C_G(xz) = C_G(x)$.
Thus precisely $a/b_x$ classes lie above $[\bar{x}]$, with identical centralizers.
This proves \eqref{eq:classlift} and \eqref{eq:exact}.

To bound the electric charges in each lifted class, note that $A \le Z(C_G(x))$.
Applying Lemma~\ref{lem:classes} to $C_G(x)$ and then to $K_x \le C_Q(\bar{x})$ of index $b_x$ gives
\begin{equation}
 k(C_G(x)) \ge k(K_x) + a - 1 \ge b_x^{-1} k(C_Q(\bar{x})) + a - 1.
\end{equation}
Multiplying by $a^2/b_x$ and summing using \eqref{eq:exact} yields \eqref{eq:weighted},
noting $a^2(a-1)\sum b_x^{-1} = a(a-1)k(G)$.
Because $(a/b_x)^2 \ge 1$ with $(a/b_1)^2 = a^2$, retaining this excess yields \eqref{eq:centralbound}.
\end{proof}

\begin{corollary}[Central towers]\label{cor:tower}
Suppose $G_0 \twoheadrightarrow G_1 \twoheadrightarrow \cdots \twoheadrightarrow G_t$ has central
kernels of orders $a_i$ in the maps $G_i \to G_{i+1}$. Setting $A_i = \prod_{j<i} a_j$ and $A_0 = 1$,
\begin{equation}\label{eq:tower}
 A_t\,r_1(G_0) \ge r_1(G_t) + \sum_{i=0}^{t-1} \frac{A_t}{A_{i+1}} \bigl((a_i^2-1)k(G_{i+1}) + a_i(a_i-1)k(G_i)\bigr).
\end{equation}
\end{corollary}
\begin{proof}
Dividing \eqref{eq:centralbound} for the $i$-th map by $A_{i+1}$ reveals that
$r_1(G_i)/A_i - r_1(G_{i+1})/A_{i+1}$ is bounded below by the displayed correction.
Summing over $i$ telescopes, proving \eqref{eq:tower}.
\end{proof}
\subsubsection{Sharpness of quotient loss: Extraspecial groups}
To test the sharpness of the quotient loss bound $r_1(Q) \le a\,r_1(G)$, we examine extraspecial $p$-groups.
An extraspecial $p$-group $P$ of order $p^{2m+1}$ has center $Z(P) \cong C_p$ equal to its commutator
subgroup $P'$, with central quotient $V = P/Z(P) \cong \mathbb{F}_p^{2m}$.
Commutation is encoded by a non-degenerate alternating bilinear form on $V$, reducing the rank calculation
to finite symplectic linear algebra.

\begin{proposition}[Double rank of extraspecial groups]\label{prop:extra}
If $P$ is extraspecial of order $p^{2m+1}$, $m \ge 1$, then
\begin{equation}\label{eq:extra}
 r_1(P) = p^{2m-1}(p^{2m} + p^3 - 1).
\end{equation}
\end{proposition}
\begin{proof}
The commutator induces a non-degenerate alternating bilinear form on $V = P/Z(P)$ of dimension $2m$ over $\mathbb{F}_p$.
Put $q = p^{2m}$. An orthogonal triple in $V$ lifts to $p^3$ commuting elements in $P$.
Counting orthogonal triples: if the first vector is zero, there are $q + (q-1)q/p$ choices for the other two.
If the first vector is non-zero ($q-1$ choices), the second lies in its orthogonal complement ($q/p$ choices);
it either lies in its span ($p$ choices, leaving $q/p$ choices for the third) or outside its span ($q/p-p$ choices, leaving $q/p^2$ choices).
Thus the total number of orthogonal triples in $V$ is
\begin{equation}
 N = q + \frac{(q-1)q}{p} + (q-1)\left[p\,\frac{q}{p} + \left(\frac{q}{p}-p\right)\frac{q}{p^2}\right] = \frac{q^2(q + p^3 - 1)}{p^3}.
\end{equation}
Multiplying by $p^3$ lifts and dividing by $|P| = pq = p^{2m+1}$ gives $r_1(P) = p^3 N / (pq) = p^{2m-1}(p^{2m} + p^3 - 1)$.
The proof applies to both extraspecial isomorphism types for all primes $p$.
\end{proof}
Formula \eqref{eq:extra} is the double-rank normalization
$r_1(P)=|P|^2P_3(P)$ of the commuting-probability formula in
Ref.~\cite[Corollary 7.4]{LSrigid}; the counting proof is included for completeness.

\begin{theorem}[Sharp loss at every kernel order]\label{thm:sharp}
For every integer $a \ge 2$, there exists a sequence of central extensions
$1 \to A_m \to G_m \to Q_m \to 1$ with $|A_m| = a$ such that
\begin{equation}
 \lim_{m \to \infty} \frac{r_1(Q_m)}{r_1(G_m)} = a.
\end{equation}
Thus, the prefactor $a$ in Theorem~\ref{thm:quotient} is strictly optimal.
\end{theorem}
\begin{proof}
For a prime $p$, $P/Z(P) \cong \mathbb{F}_p^{2m}$, so
\begin{equation}
 \frac{r_1(P/Z(P))}{r_1(P)} = \frac{p^{4m}}{p^{2m-1}(p^{2m} + p^3 - 1)} \longrightarrow p \quad \text{as } m \to \infty.
\end{equation}
For general $a = \prod_{j=1}^s p_j$, taking $G_m$ to be the direct product of extraspecial groups
of orders $p_j^{2m+1}$ and $A_m$ the product of their centers yields the limit $a$ by rank multiplicativity.
\end{proof}

While extraspecial groups saturate the local quotient loss ratio $r_1(Q)/r_1(G) \to a$,
they are inefficient candidates for maximizing $\cD$ at bounded rank: for fixed $p$,
$\log r_1(P) \sim 4m\log p$ while $\log|P| \sim 2m\log p$, giving $2\log|P|/(\log r_1(P))^2 \to 0$.
Local sharpness of an extension bound and global optimality of a family envelope
are fundamentally distinct questions.

\subsection{Solvable and radical-free refinements}
\label{subsec:solvable}

Solvable gauge groups provide a clear demonstration of the difference between fixing
an algebraic structural parameter and allowing it to grow.
Recall that the derived series of $G$ is $G^{(0)} = G$ and $G^{(i+1)} = [G^{(i)}, G^{(i)}]$;
$G$ is solvable of derived length $d$ if $G^{(d)} = 1$.
Nilpotency class 2 is the stronger condition $G' \le Z(G)$, meaning that all commutators are central.

When derived length or nilpotency class is bounded, polynomial class-number bounds
\cite{Bertram} yield explicit power-law lower bounds on categorical rank.

\begin{proposition}[Rank bounds for bounded derived length]\label{prop:derived}
If $G$ is solvable of derived length at most $d \ge 1$, then
\begin{equation}
 r_1(G) \ge |G|^{\alpha_d}, \qquad
 \alpha_1 = 2, \qquad
 \alpha_d = (3 \cdot 2^{d-2} - 1)^{-1} \quad (d \ge 2).
\end{equation}
If $G$ has nilpotency class at most 2, then $r_1(G) \ge |G|$.
Consequently,
\begin{equation}
 F_{\mathrm{dl}\le d}^{(1)}(R) \le 2\alpha_d^{-1}\log R, \qquad
 F_{\mathrm{class}\le 2}^{(1)}(R) \le 2\log R.
\end{equation}
\end{proposition}
\begin{proof}
For Abelian groups ($d=1$), $r_1(G) = |G|^2$, so $\alpha_1 = 2$.
For $d \ge 2$, set $N = G'$ and $b = |G/N|$.
The identity-flux sector alone gives $r_1(G) \ge k(G) \ge |G/G'| = b$.
On the other hand, $N$ has derived length at most $d-1$ and order $|G|/b$.
By subgroup monotonicity (Lemma~\ref{lem:basic}) and induction,
\begin{equation}
 r_1(G) \ge \frac{r_1(N)}{b} \ge \frac{(|G|/b)^{\alpha_{d-1}}}{b} = \frac{|G|^{\alpha_{d-1}}}{b^{\alpha_{d-1}+1}}.
\end{equation}
Balancing these two bounds across all possible values of $b > 0$ reveals that the worst case
occurs when $b = |G|^{\alpha_{d-1}/(\alpha_{d-1}+2)}$, giving $r_1(G) \ge |G|^{\alpha_d}$ with
$\alpha_d = \alpha_{d-1}/(\alpha_{d-1}+2)$, which solves to $\alpha_d = (3 \cdot 2^{d-2} - 1)^{-1}$.
For class 2 groups ($G' \le Z(G)$), the central fluxes alone yield
$r_1(G) \ge \sum_{z \in Z(G)} k(C_G(z)) = |Z(G)|k(G) \ge |Z(G)|\,|G/G'| \ge |G|$.
Inverting $r_1 \le R$ proves the envelope bounds.
\end{proof}

When the derived length is allowed to grow arbitrarily, the exponent $\alpha_d$ tends to zero.
To obtain uniform bounds across all solvable groups, we leverage Keller's class-number theorem \cite{Keller}.

\begin{theorem}[Uniform solvable envelopes]\label{thm:solvable}
There exists an absolute constant $C > 0$ such that for all $R \ge 1$,
\begin{align}
 F_{\mathrm{solv}}(R) &\le C R\log(2+R) \quad (\text{arbitrary twists}),
 \label{eq:solvable-twisted}\\
 F_{\mathrm{solv}}^{(1)}(R) &\le C R\log(2 + \log(2+R)) \quad (\text{untwisted}).
 \label{eq:solvable-untwisted}
\end{align}
\end{theorem}
\begin{proof}
Keller \cite[Theorem A]{Keller} proved that every solvable group $H$ of
sufficiently large order satisfies
\begin{equation}
 k(H) \ge c\,\frac{\log|H|}{\log\log|H|},
 \label{eq:keller}
\end{equation}
for an absolute constant $c > 0$. We use the logarithmic estimates below
only in this large-order regime; smaller groups are absorbed into the
constants in the regularized final bounds.
In any twisted theory, the identity flux sector has trivial cocycle $\alpha_1 = 1$,
giving $r_\omega(G) \ge k(G)$.
Setting $L = \log|G|$, Keller's theorem implies $r_\omega(G) \ge c L/\log L$.
Inverting this relation yields $L \le C r_\omega(G)\log(2 + r_\omega(G))$, which establishes
\eqref{eq:solvable-twisted} after taking the supremum over $r_\omega(G) \le R$.

In the untwisted case, we apply Keller's bound inside the centralizers of class representatives.
Let $k = k(G)$ and let $m_i$ denote the orders of the centralizers $C_G(x_i)$.
The class equation $\sum_{i=1}^k m_i^{-1} = 1$ implies that at least $k/2$ centralizers
satisfy $m_i \ge k/2$.
Because subgroups of solvable groups are solvable, applying \eqref{eq:keller} to these
large centralizers gives
\begin{equation}
 r_1(G) = \sum_{i=1}^k k(C_G(x_i)) \ge \frac{k}{2} \cdot c\,\frac{\log(k/2)}{\log\log(k/2)}
 \ge c' k\,\frac{\log k}{\log\log k}.
\end{equation}
If $k \ge L$, then $r_1(G) \ge L$ trivially.
If $k < L$, combining with $k \ge c L/\log L$ yields $\log k = \log L + O(\log\log L)$
and $\log\log k = \log\log L + o(1)$.
Substituting gives $r_1(G) \ge c'' L/\log\log L$.
Inverting this relation yields $L \le C r_1(G)\log(2 + \log(2 + r_1(G)))$, proving
\eqref{eq:solvable-untwisted}.
\end{proof}

Turning to the opposite structural regime, consider \emph{radical-free} groups ($\mathrm{Rad}(G) = 1$).
These groups possess no non-trivial solvable normal subgroups, and their socles are direct
products of non-Abelian simple groups.
Using deep class-number bounds of Baumeister, Mar\'oti, and Tong-Viet \cite{BMTV},
we establish a polylogarithmic envelope for radical-free gauge groups under arbitrary twists.

\begin{corollary}[Radical-free polylogarithmic envelope]\label{cor:radicalfree}
For every $\varepsilon > 0$, there exists a constant $C_\varepsilon > 0$ such that,
allowing arbitrary 3-cocycle twists,
\begin{equation}
 F_{\mathrm{Rad}=1}(R) \le C_\varepsilon (\log R)^{3+\varepsilon}.
 \label{eq:radicalfree}
\end{equation}
Furthermore, the envelope satisfies the explicit non-asymptotic bound
\begin{equation}
 F_{\mathrm{Rad}=1}(R) \le 2\log\left(3^{\lfloor(R+2)/3\rfloor} - 1\right)
 < \frac{2}{3}(R+2)\log 3.
 \label{eq:radicalfree-explicit}
\end{equation}
\end{corollary}
\begin{proof}
Ref.~\cite[Theorem 2.1]{BMTV} establishes that for any $\varepsilon > 0$, there exists
$\delta_\varepsilon > 0$ such that every non-trivial radical-free finite group satisfies
\begin{equation}
 \log k(G) > \delta_\varepsilon (\log|G|)^{1/(3+\varepsilon)}.
\end{equation}
Because $k(G) \le r_\omega(G) \le R$ for every twist $\omega$, raising this inequality
to the power $3+\varepsilon$ immediately yields
$2\log|G| \le C_\varepsilon (\log R)^{3+\varepsilon}$, establishing \eqref{eq:radicalfree}.

For the explicit bound, Ref.~\cite[Theorem 1.2]{BMTV} proves that every radical-free
group satisfies $|G| < 3^{k(G)}$.
By Corollary~\ref{cor:classes}, $r_\omega(G) \ge 3k(G) - 2$, which implies
$k(G) \le \lfloor(R+2)/3\rfloor$ whenever $r_\omega(G) \le R$.
Because $|G|$ is an integer, $|G| \le 3^{\lfloor(R+2)/3\rfloor} - 1$.
Taking the logarithm proves \eqref{eq:radicalfree-explicit}.
\end{proof}

\subsection{Global finiteness and the remaining gap}
\label{subsec:global}

Finally, we address the global finiteness of Dijkgraaf--Witten envelopes across
\emph{all} finite groups.
By combining the large-centralizer iteration method of Erd\H{o}s and Straus \cite{ES}
with general class-number bounds \cite{BMTV}, we establish that the search space
of DW theories at any rank cutoff is rigorously finite.

\begin{theorem}[Global finiteness envelopes]\label{thm:global}
For every $\varepsilon > 0$, there exists a constant $C_\varepsilon > 0$ such that
every finite group $G$ satisfies
\begin{equation}
 \log|G| \le C_\varepsilon r_1(G)\,[\log(2 + r_1(G))]^{2+\varepsilon}.
 \label{eq:global}
\end{equation}
Consequently, the untwisted and twisted global envelopes satisfy
\begin{align}
 F_{\DW}^{(1)}(R) &\le C_\varepsilon R[\log(2+R)]^{2+\varepsilon},
 \label{eq:global-untwisted}\\
 F_{\DW}(R) &\le C_\varepsilon R[\log(2+R)]^{3+\varepsilon}.
 \label{eq:global-twisted}
\end{align}
\end{theorem}
\begin{proof}
By Ref.~\cite[Theorem 1.1]{BMTV}, for every $\eta > 0$ there exists $c_\eta > 0$ such that
every finite group $K$ of sufficiently large order satisfies
\begin{equation}
 k(K) \ge c_\eta \frac{\log|K|}{(\log\log|K|)^{3+\eta}}.
\end{equation}
All uses of this estimate are in the large-order regime; the finitely many
smaller groups are absorbed by enlarging the final constants.
Applying this to the large centralizers of $G$: let $k = k(G)$ and let $m_i = |C_G(x_i)|$.
By the class equation $\sum_{i=1}^k m_i^{-1} = 1$, at least $k/2$ centralizers satisfy
$m_i \ge k/2$. Summing their class numbers yields
\begin{equation}
 r_1(G) = \sum_{i=1}^k k(C_G(x_i)) \ge c'_\eta k \frac{\log k}{(\log\log k)^{3+\eta}}.
 \label{eq:iteration}
\end{equation}
Letting $L = \log|G|$, if $k \ge L$ then $r_1(G) \ge L$ suffices.
Otherwise, $k \ge c_\eta L/(\log L)^{3+\eta}$, giving $\log k = \log L + O(\log\log L)$.
Substituting into \eqref{eq:iteration} gives
\begin{equation}
 r_1(G) \ge c''_\eta \frac{L}{(\log L)^{2+\eta}(\log\log L)^{3+\eta}}.
\end{equation}
Setting $\eta = \varepsilon/4$ and absorbing the iterated logarithm yields
$r_1(G) \ge c L / (\log L)^{2 + 3\varepsilon/4}$.
Inverting this relation gives $L \le C_\varepsilon r_1(G)[\log(2+r_1(G))]^{2+\varepsilon}$,
which proves \eqref{eq:global} and \eqref{eq:global-untwisted}.

For twisted theories, we use $r_\omega(G) \ge k(G)$ directly with the general bound
$k(G) \ge c_\varepsilon L/(\log L)^{3+\varepsilon}$.
Inverting this relation yields $L \le C_\varepsilon r_\omega(G)[\log(2+r_\omega(G))]^{3+\varepsilon}$,
establishing \eqref{eq:global-twisted}.
\end{proof}

Combining the lower bound from permutation gauge groups (Corollary~\ref{cor:symmetric})
with the global upper bound (Theorem~\ref{thm:global}) frames the known boundaries of the
untwisted Dijkgraaf--Witten envelope:
\begin{equation}
 \left(\frac{2}{\sqrt{3\zeta(2)\zeta(3)}} + o(1)\right)(\log R)^{3/2}\log\log R
 \le F_{\DW}^{(1)}(R)
 \le C_\varepsilon R[\log(2+R)]^{2+\varepsilon}.
 \label{eq:main-envelope-bounds}
\end{equation}
Closing this substantial gap—and determining whether $F_{\DW}^{(1)}(R)$ scales as
$(\log R)^2$—remains one of the most intriguing open problems in the classification
of discrete gauge theories.

\section{Conclusion and Outlook}
\label{sec:outlook}

In this work, we have investigated the fundamental trade-off between the total
quantum dimension $\cD$ and categorical rank $R$ in $(2+1)$-dimensional
Dijkgraaf--Witten topological field theories.
Physically, this optimization governs the maximal state-optimized torus topological
entanglement entropy $\max_\psi \Gamma_{T^2}(\psi) = 2\log\cD$ achievable at a
prescribed torus ground-state degeneracy bound $R$, while quantitatively constraining
the search space required for the classification of topological orders.

Our results demonstrate how topological twists (3-cocycles) and gauge-group
architectures independently dictate this trade-off:
\begin{enumerate}
 \item \emph{Topological twist compression.}
 For a fixed gauge group $G$, $\cD = |G|$ is topologically rigid, but non-trivial
 3-cocycles $\omega$ can compress the anyon spectrum by projecting out non-regular
 holonomies and inducing projective electric representations.
 We identified the geometric mechanism of \emph{cyclic protection}: commuting
 holonomies generating a cyclic subgroup possess trivial 3-holonomy phases
 $\Omega_\omega \equiv 1$ under every twist, establishing a universal floor
 $\Xi(G) \le r_\omega(G)$.
 For elementary-Abelian gauge groups $(\mathbb{Z}_p)^n$ at fixed prime $p$,
 an exact first-moment analysis over the space of alternating trilinear forms
 revealed that optimal twists compress the spectrum to linear scaling $m_p(n) = \Theta_p(p^n)$,
 elevating the envelope from $F_p^{(1)}(R) = \log R + O_p(1)$ to $F_p(R) = 2\log R + O_p(1)$
 and doubling the growth rate of $\cD_{\max}(R)$ from square-root to linear.
 \item \emph{Gauge-group structural envelopes.}
 In the untwisted setting, the envelope is shaped by group-theoretic structure.
 For permutation gauge groups $S_n$ and $A_n$, the classification of $\mathbb{Z}^2$-sets
 leads to the sharp asymptotic envelope $F_S^{(1)}(R) \sim F_A^{(1)}(R) \sim A_S (\log R)^{3/2}\log\log R$.
 For almost-simple groups and their finite direct products, we proved an absolute
 quadratic-logarithmic envelope upper bound $F_{\mathrm{a.s.}}^{(1)}(R) \le C(\log R)^2$.
 For solvable groups, the general bounds are $O(R\log R)$ with twisting and
 $O(R\log\log R)$ without twisting; the radical-free bound is
 $O_\varepsilon((\log R)^{3+\varepsilon})$ for arbitrary twists.
 Quotient-loss theorems and extraspecial enumerations quantify the rank cost
 of central extensions.
\end{enumerate}

It is illuminating to contrast these Dijkgraaf--Witten scales with chiral conformal
field theories. In WZW modular tensor categories, the envelope in binary logarithmic
units $q_R = \log_2 R$ exhibits quadratic growth \cite{ShenWZW}:
\begin{equation}
 \lim_{R \to \infty} \frac{F_{\WZW}(R)}{q_R^2} = \frac{7\zeta(3)}{4\pi^2} \approx 0.2131.
 \label{eq:WZW-comparison}
\end{equation}
In contrast, permutation gauge theories yield $F_S^{(1)}(R) = O(q_R^{3/2}\log q_R) = o(q_R^2)$,
and elementary-Abelian twisted theories scale linearly as $O(q_R)$.
Both explicit families therefore lie strictly below the WZW quadratic scale.

The upper-bound direction of a possible quadratic-logarithmic law is
formulated in the following conjecture. Attainment of that scale is a
separate open problem.

\begin{conjecture}[Universal quadratic-logarithmic envelope]\label{conj:upper}
There exists an absolute constant $C > 0$ such that
\begin{equation}
 \log|G| \le C(\log r_1(G))^2
\end{equation}
for every finite group $G$.
Equivalently, the untwisted Dijkgraaf--Witten envelope satisfies $F_{\DW}^{(1)}(R) = O((\log R)^2)$.
\end{conjecture}

Establishing Conjecture~\ref{conj:upper} would confirm that Dijkgraaf--Witten theories
obey the same universal quadratic-logarithmic ceiling observed in chiral WZW models.
To achieve a matching lower envelope $F_{\DW}^{(1)}(R) = \Omega((\log R)^2)$, one would
require a family of groups $\{G_n\}$ exhibiting positive quadratic efficiency:
\begin{equation}
 C(G_\bullet) := \limsup_{n \to \infty} \frac{2\log|G_n|}{(\log r_1(G_n))^2} > 0.
\end{equation}
In the WZW setting, this balance is achieved because the sector count is exponential
in a linear parameter $t$, while $\log\cD$ grows quadratically as $t^2$ \cite{ShenWZW}.
For a finite gauge group, this would require $\log r_1(G_t) \asymp t$ while
$\log|G_t| \asymp t^2$.
As shown in Section~\ref{sec:structure} and Appendix~\ref{app:technical}, the
families analyzed here—bounded derived length, nilpotency class at most two,
unitriangular and extraspecial groups, symmetric and alternating groups,
fixed-base symmetric wreath products, and the specified prime-cyclic wreath
tower—have vanishing quadratic efficiency along sequences of growing order.
This does not rule out general varying-base wreath products, affine towers,
branch-like groups, or mixed constructions.

Several promising avenues emerge from this work:
\begin{itemize}
 \item \emph{Deterministic cocycle constructions.}
 While Theorem~\ref{thm:mean} proves the existence of alternating trilinear forms
 with linear rank $O_p(p^n)$ via the probabilistic method, identifying explicit,
 efficiently specified deterministic families of forms achieving this optimal scaling remains an open challenge
 in finite geometry.
 \item \emph{Higher-holonomy protected configurations.}
 The cyclic protection index $\Xi(G)$ provides a rigorous lower bound by considering
 pairs of holonomies generating cyclic subgroups.
 Investigating whether non-cyclic commuting configurations can be protected by
 higher-cohomological obstructions could significantly sharpen the lower bound on
 twisted rank.
 \item \emph{Extremal non-solvable architectures.}
 Exploring sophisticated finite group architectures—such as varying-base wreath
 products, automorphism towers of simple groups, and non-split extensions—may provide
 candidate families with positive quadratic efficiency, potentially closing the gap
 between the permutation lower bound and the global upper bound.
\end{itemize}

Ultimately, determining the maximal total quantum dimension at bounded rank provides
a rigorous bridge between finite group theory, category theory, and topological quantum
matter, shedding light on the universal entanglement capacity of quantum phases.

\acknowledgments
The author is supported by BIMSA and the NSFC under Grant No. 12505090.
OpenAI GPT6-astra was used to assist with literature searches, earlier drafting,
mathematical derivations and proof checks, and preparation of computational
verification.

\appendix
\section{Technical Details and Further Gauge Families}
\label{app:technical}
\label{app:architectures}

\subsection{Alternation of the three-holonomy phase}
\label{app:alternation}

Here we provide the detailed algebraic verification of Lemma~\ref{lem:alternation},
confirming that the 3-holonomy phase $\Omega_\omega$ defines an alternating
tricharacter on any Abelian subgroup $A \le G$, independent of field characteristic.

\begin{proof}
The normalized cocycle identity is
\[
 \omega(b,c,d)\omega(a,bc,d)\omega(a,b,c)
 =\omega(ab,c,d)\omega(a,b,cd).
\]
For commuting $a,b,c,d$, let $R(a,b,c,d)$ be the left-hand side divided
by the right-hand side. Put
\[
 P_+=\{abdc,acbd,adcb,cadb,dabc,dcab\},\qquad
 P_-=\{abcd,acdb,adbc,cabd,cdab,dacb\},
\]
where a word specifies the ordered arguments of $R$. Expansion and
cancellation give
\[
 \frac{\Omega_\omega(ab,c,d)}
 {\Omega_\omega(a,c,d)\Omega_\omega(b,c,d)}
 =\frac{\prod_{w\in P_+}R(w)}{\prod_{w\in P_-}R(w)}=1.
\]
Cyclic symmetry gives multiplicativity in the other slots:
$\Omega_\omega(g,h,x) = \Omega_\omega(h,x,g) = \Omega_\omega(x,g,h)$.
A transposition inverts \eqref{eq:omega}, and repeated arguments cancel.
Substituting $\delta\beta(a,b,c)=\beta(b,c)\beta(a,bc)/[\beta(ab,c)\beta(a,b)]$
in \eqref{eq:omega} cancels all factors of $\beta$, ensuring that $\Omega_\omega$
depends only on the cohomology class $[\omega]|_A \in H^3(A,\U)$.
\end{proof}

\subsection{Finite-rank census and verification}
\label{app:census}
To validate the contraction formula \eqref{eq:contraction} and singular-line
dictionary \eqref{eq:lines}, Table~\ref{tab:census} reports the exact rank
distribution obtained from an exhaustive enumeration of alternating trilinear
forms over finite fields $\mathbb{F}_p^n$.
Forms are parameterized by their $\binom{n}{3}$ coefficients in a fixed basis
(for orbit classifications under $\mathrm{GL}_n(\mathbb{F}_p)$, see Ref.~\cite{HoraPudlak}).
The parameter ranges are $3\le n\le5$ for $p=2$, $3\le n\le4$ for $p=3$,
and $n=3$ for $p=5$, covering characteristic two and two small odd primes.
For each pair $(p,n)$, all $p^{\binom n3}$ forms are evaluated at all $p^n$
fluxes. The largest listed coefficient space contains $2^{10}=1024$ forms;
already the next binary dimension has $2^{20}$ forms and $2^6$ fluxes per form.
The table is exhaustive within each listed parameter pair, rather than across
all elementary-Abelian groups below a common order cutoff. The dimension-6
field-trace construction is checked separately and is not a full census.

\begin{table}[htbp]
\centering
\caption{Exhaustive rank census of alternating trilinear forms on $\mathbb{F}_p^n$.
The notation $r:N$ denotes that exactly $N$ distinct forms achieve categorical rank $r$.}
\label{tab:census}
\begin{tabular}{ccl}
\toprule
$p$&$n$&Rank distribution $r:N$\\
\midrule
2&3&$22:1,\quad64:1$\\
2&4&$88:15,\quad256:1$\\
2&5&$184:868,\quad352:155,\quad1024:1$\\
3&3&$105:2,\quad729:1$\\
3&4&$945:80,\quad6561:1$\\
5&3&$745:4,\quad15625:1$\\
\bottomrule
\end{tabular}
\end{table}
These exact evaluations corroborate the first-moment averages \eqref{eq:mean}
and the parity obstructions (Proposition~\ref{prop:parity}), and determine
the minimal ranks in the listed dimensions.

\subsection{Abelian subgroups and unitriangular groups}

A large Abelian subgroup already supplies many commuting triples.
This simple observation can rule out a candidate family before its
full conjugacy-class structure is computed.
This and the next two subsections test distinct ways of increasing group
complexity: growing nilpotency class in unitriangular groups, an Abelian
normal subgroup acted on by a complement, and repeated cyclic wreath
extensions. The first and third yield exclusions of specific candidate
families; the affine calculation isolates a contribution to rank that can
be used to test a proposed semidirect product.

\begin{lemma}\label{lem:abelian}
For any Abelian subgroup $B\leq G$,
$r_1(G)\geq|B|^3/|G|$.
\end{lemma}
\begin{proof}
All triples in $B^3$ commute. Count them inside $h_3(G)$ and divide by
$|G|$.
\end{proof}
The unitriangular group $\mathrm{UT}_n(q)$ consists of upper-triangular
$n\times n$ matrices over $\mathbb{F}_q$ with ones on the diagonal.
Its order is $|\mathrm{UT}_n(q)| = q^{n(n-1)/2}$.
While its nilpotency class $n-1$ grows with matrix size, the presence of large
Abelian subgroups forces the rank to grow with the same quadratic scaling.

\begin{corollary}\label{cor:UT}
For $n\ge 2$ and prime power $q$,
\[
 |\mathrm{UT}_n(q)|=q^{n(n-1)/2},\qquad
 r_1(\mathrm{UT}_n(q))\geq
 q^{3\lfloor n^2/4\rfloor-n(n-1)/2}.
\]
In particular, $\log r_1(\mathrm{UT}_n(q))=\Theta(n^2\log q)$
as $n\to\infty$, uniformly for $q\geq2$.
\end{corollary}
\begin{proof}
For $b=\lfloor n/2\rfloor$, matrices $I+X$ with $X$ supported only
in rows $1,\ldots,b$ and columns $b+1,\ldots,n$ form an Abelian
subgroup of order $q^{b(n-b)}$: every product of two such $X$ is zero.
Apply Lemma~\ref{lem:abelian}. The exponent is
$n^2/4+n/2+O(1)$. The opposite order bound is the trivial
$r_1(G)\leq|G|^2$.
\end{proof}
Thus, increasing nilpotency class alone through unitriangular groups does not yield
sub-quadratic rank growth.

\subsection{A split-affine sector diagnostic}
For a semidirect product with Abelian normal subgroup, the question is
whether the complement can combine enough flux and charge labels into
orbits to keep the rank small. Counting just the fluxes in the normal
subgroup gives a necessary constraint before the other sectors are analyzed.
Let $H$ act by automorphisms on a finite Abelian group $V$ (written additively),
forming the semidirect product $G = V \rtimes H$ with multiplication
$(v,h)(w,k) = (v + h\cdot w, hk)$.
Let $\widehat{V} = \Hom(V,\U)$ denote the Pontryagin dual group.
For any subgroup $K \le H$, let $V^K$ and $\widehat{V}^K$ denote the fixed points
in $V$ and $\widehat{V}$, respectively.
We define the contribution of flux sectors originating in $V$ by
\[
 T_V(G)=\sum_{[v]_G\subset V}k(C_G(v))\leq r_1(G).
\]
\begin{proposition}\label{prop:affine}
The flux contribution $T_V(G)$ evaluates to
\begin{equation}\label{eq:affine}
 T_V(G)=\frac1{|H|}\sum_{\substack{x,y\in H\\xy=yx}}
 |V^{\langle x,y\rangle}|\,|\widehat V^{\langle x,y\rangle}|.
\end{equation}
Consequently $T_V(G)\geq |V|^2/|H|$, $T_V(G)\geq k(H)$, and
$T_V(G)\geq n_H(V)^2$, where $n_H(V)$ is the number of $H$-orbits on $V$.
\end{proposition}
\begin{proof}
For $v\in V$, $C_G(v)=V\rtimes H_v$. The little-group form of
Clifford theory gives
\[
 T_V(G)=\sum_{[(v,\lambda)]\in(V\times\widehat V)/H}k(H_{v,\lambda}).
\]
Here $\lambda$ extends to $V\rtimes(H_v)_\lambda$ by setting its value
on the stabilizer to $1$, so no projective multiplier occurs. For any
finite $H$-set $X$, Burnside's lemma gives
\[
 \sum_{[z]\in X/H}k(H_z)
 =\frac1{|H|}\sum_{xy=yx}|X^{\langle x,y\rangle}|.
\]
Both sides count $H$-orbits of pairs $(z,x)$ with $xz=z$.
Apply this to $X=V\times\widehat V$. The identity pair gives
$|V|^2/|H|$; the orbit of $(0,1)$ gives $k(H)$. Retaining only
$(x,1)$ in the Burnside sum gives $|H|^{-1}\sum_x|V^x|^2$, because
an automorphism of a finite Abelian group has equally many fixed vectors
and fixed characters. Cauchy--Schwarz and Burnside's lemma give the
last bound.
\end{proof}
For a cyclic complement $H=C_p$ acting on $V=\mathbb F_p^m$, write
$f=\dim V^{C_p}$. Every non-identity element generates the complement, so
\eqref{eq:affine} specializes to $T_V=p^{2m-1}+(p^2-1)p^{2f-1}$.
This specialization displays the rank cost of the fixed subspace: its
contribution grows as $p^{2f}$ at fixed $p$, even before other fluxes are counted.
The full rank still includes fluxes outside $V$.

\subsection{Repeated prime-cyclic wreath products}

Consider the iterated wreath product $W_{d+1} = W_d \wr C_p = W_d^p \rtimes C_p$,
where $C_p$ cyclically permutes $p$ copies of $W_d$, initialized by $W_1 = C_p$.
These groups arise naturally as Sylow $p$-subgroups of symmetric groups $S_{p^d}$.
This tower tests whether repeated cyclic extensions, beyond a fixed number
of extension steps, can yield positive quadratic efficiency. The exact
recursion below shows why this particular tower fails that test.
The local-system counting method is related to the general wreath-product
enumeration of Ref.~\cite{FS}. For related commuting-pair probabilities with
Abelian base and acting groups, see Ref.~\cite{ErovenkoSury}; the recursion
below instead concerns the double rank, or commuting-triple normalization.

\begin{proposition}[Wreath product recursion]\label{prop:cyclicwreath}
For every finite $H$ and prime $p$, with $C_p$ acting transitively on $p$ coordinates,
\begin{equation}\label{eq:cyclicwreath}
 r_1(H\wr C_p)=\frac{r_1(H)^p+(p^3-1)r_1(H)}{p}.
\end{equation}
\end{proposition}
\begin{proof}
Write $r=r_1(H)$ and classify homomorphisms from $\Gamma=\mathbb Z^2$
to $H^p\rtimes C_p$ up to conjugacy. If the projected map to $C_p$
is trivial, conjugation by $H^p$ leaves a $p$-tuple of commuting-pair
orbits of $H$. Rotation by $C_p$ has $(r^p+(p-1)r)/p$ orbits on these
tuples, by Burnside's lemma.

There are $p^2-1$ nontrivial projected maps $\rho:\Gamma\to C_p$;
each is transitive on the coordinates. Fixing a base point identifies
its stabilizer with $L=\ker\rho$, an index-$p$ subgroup of $\mathbb Z^2$.
Equivariant $H$-labels on this transitive orbit, modulo changes of labels
at the $p$ points, are homomorphisms $L\to H$ up to conjugacy.
This can be seen by choosing coset representatives of $L$ in $\Gamma$:
transport along them fixes labels away from the base point, leaving
exactly the return transports given by the homomorphism on $L$.
Since $L\cong\mathbb Z^2$, there are $r$ choices. A rotation changes
the base point but induces only conjugation on this return data;
conjugation of $L$ by $\Gamma$ is trivial because $\Gamma$ is Abelian.
Thus each nontrivial $\rho$ contributes $r$ orbits. Adding the two
contributions proves \eqref{eq:cyclicwreath}.
\end{proof}

Applying Proposition~\ref{prop:cyclicwreath} to the Sylow tower $W_1 = C_p$, $W_{d+1} = W_d \wr C_p$:
\begin{align*}
 \log|W_d|&=\frac{p^d-1}{p-1}\log p,\\
 \log r_1(W_d)&\geq
 \frac{(2p-3)p^{d-1}+1}{p-1}\log p.
\end{align*}
Together with $r_1(W_d)\leq|W_d|^2$, this proves $\log r_1(W_d)=\Theta_p(\log|W_d|)$,
so this tower has quadratic efficiency $C(W_\bullet)=0$.

\subsection{Envelope constants and variational formulation}
\label{app:envelope-constants}

Here we establish the formal variational identity connecting the asymptotic envelope
to sequence efficiencies.

\begin{proposition}[Variational limsup identity]\label{prop:variational}
Allowing infinite values,
\begin{equation}\label{eq:variational}
 \limsup_{R\to\infty}\frac{F_{\DW}^{(1)}(R)}{(\log R)^2}
 =\sup_{\{G_n\}} C(G_\bullet),
\end{equation}
where the supremum on the right-hand side is taken over all sequences of finite groups
$\{G_n\}$ with $|G_n| \to \infty$.
\end{proposition}
\begin{proof}
Theorem~\ref{thm:global} implies that growing order forces growing rank.
For any sequence, $F_{\DW}^{(1)}(r_1(G_n))\geq2\log|G_n|$. Taking limsups
shows that the left-hand side of \eqref{eq:variational} is at least
$C(G_\bullet)$ for every family, hence at least their supremum.
Conversely, there are finitely many isomorphism classes
of groups of bounded order, and the same theorem bounds order at each
rank budget. Hence choose an actual maximizing group $G_R$ for $F_{\DW}^{(1)}(R)$.
Along any sequence $R_j\to\infty$ realizing the left-hand limsup,
$|G_{R_j}|\to\infty$ because $F_{\DW}^{(1)}(R)\to\infty$ by the symmetric lower
bound. Since $r_1(G_{R_j})\leq R_j$,
\[
 \frac{2\log|G_{R_j}|}{(\log r_1(G_{R_j}))^2}
 \geq\frac{F_{\DW}^{(1)}(R_j)}{(\log R_j)^2}.
\]
This proves the reverse inequality.
\end{proof}

A family giving a positive efficiency only along a sparse sequence
does not automatically yield a lower bound at every large budget.
For example, a sufficient interpolation condition is a sequence of
increasing ranks $r_n$ with bounded ratios $\log r_{n+1}/\log r_n$ and
$2\log|G_n|\geq c(\log r_n)^2$. Choosing the last admissible term then
gives $F_{\DW}^{(1)}(R)=\Omega((\log R)^2)$. If the ratios tend to $1$ and the
efficiencies converge to $C>0$, this argument instead yields
$\liminf F_{\DW}^{(1)}(R)/(\log R)^2\geq C$.

\subsection{Additional computational validations}
\label{app:checks}

The following finite cases and symbolic identities were checked. These tests
support the derivations but do not replace proofs of the asymptotic statements:
\begin{itemize}
 \item \emph{Group extensions and central lifting.}
 The quotient loss bound (Theorem~\ref{thm:quotient}), central lifting formula, and
 weighted correction (Theorem~\ref{thm:central}) were systematically verified for all
 144 non-isomorphic groups of order $|G| \le 32$, testing all 2,531 normal subgroups
 and all 1,395 central subgroups.
 \item \emph{Extraspecial and wreath families.}
 The extraspecial double-rank formula \eqref{eq:extra} was checked for
 $p\in\{2,3,5\}$ and $m\in\{1,2\}$; its intermediate triple-count
 identity was also checked for $p\in\{2,3,5,7\}$ and $1\le m\le3$.
 The cyclic wreath recursion \eqref{eq:cyclicwreath} was verified against direct
 conjugacy-class calculations in GAP.
 \item \emph{Alternating forms and singular lines.}
 The rank distributions in Table~\ref{tab:census}, the 6-dimensional line-spread
 construction (Proposition~\ref{prop:six}), and the first-moment formula \eqref{eq:mean}
 were evaluated via exhaustive finite-field linear algebra.
 \item \emph{Permutation generating functions.}
 The Euler product coefficients $a_n = r_1(S_n)$ were computed as exact integers through
 degree $n = 1000$ using recurrence relations for $\sigma_1(m)$.
\end{itemize}

\bibliographystyle{JHEP}
\bibliography{references}

\end{document}